\documentclass[a4paper,12pt]{article}

\usepackage{latexsym}
\usepackage{amsfonts,amsmath,amsthm,amssymb}
\usepackage{graphics}
\usepackage{graphicx,color}
\usepackage{bm}
\theoremstyle{definition}

\begin{document}
\baselineskip=20pt

\theoremstyle{plain}
\newtheorem{definition}{Definition}
\newtheorem{theorem}{Theorem}
\newtheorem{lemma}{Lemma}
\newtheorem{corollary}{Corollary}

\renewcommand{\thedefinition}{\thesection.\arabic{definition}}
\renewcommand{\thetheorem}{\thesection.\arabic{theorem}}
\renewcommand{\thelemma}{\thesection.\arabic{lemma}}
\renewcommand{\thecorollary}{\thesection.\arabic{corollary}}

\renewcommand{\vec}[1]{{\bm{#1}}}

\newcommand{\indep}{\mathop{\perp\!\!\perp}}

\renewcommand{\theequation}{\thesection.\arabic{equation}}

\newcommand{\eref}[1]{(\ref{#1})}

\newcommand{\lref}[1]{Lemma \ref{#1}}

\newcommand{\tref}[1]{Theorem \ref{#1}}

\def\yotte{\mbox{.\raisebox{1ex}{.}.}\,}
\def\yueni{\mbox{\raisebox{1ex}{.}.\raisebox{1ex}{.}}\,}

\newcommand{\Maru}[1]{\raise0.125ex\hbox{\textcircled{\small{#1}}}}

\def\stack#1#2{\stackrel{#1}{#2}}
\def\st{\stack}

\def\0{\vec{{\rm 0}}}
\def\1{\vec{{\rm 1}}}

\def\va{\vec{a}}
\def\vb{\vec{b}}
\def\vc{\vec{c}}
\def\vd{\vec{d}}
\def\ve{\vec{e}}
\def\vf{\vec{f}}
\def\vg{\vec{g}}
\def\vh{\vec{h}}
\def\vi{\vec{i}}
\def\vj{\vec{j}}
\def\vk{\vec{k}}
\def\vl{\vec{l}}
\def\vm{\vec{m}}
\def\vn{\vec{n}}
\def\vo{\vec{o}}
\def\vp{\vec{p}}
\def\vq{\vec{q}}
\def\vr{\vec{r}}
\def\vs{\vec{s}}
\def\vt{\vec{t}}
\def\vu{\vec{u}}
\def\vv{\vec{v}}
\def\vw{\vec{w}}
\def\vx{\vec{x}}
\def\vy{\vec{y}}
\def\vz{\vec{z}}

\def\vA{\vec{A}}
\def\vB{\vec{B}}
\def\vC{\vec{C}}
\def\vD{\vec{D}}
\def\vE{\vec{E}}
\def\vF{\vec{F}}
\def\vG{\vec{G}}
\def\vH{\vec{H}}
\def\vI{\vec{I}}
\def\vJ{\vec{J}}
\def\vK{\vec{K}}
\def\vL{\vec{L}}
\def\vM{\vec{M}}
\def\vN{\vec{N}}
\def\vO{\vec{O}}
\def\vP{\vec{P}}
\def\vQ{\vec{Q}}
\def\vR{\vec{R}}
\def\vS{\vec{S}}
\def\vT{\vec{T}}
\def\vU{\vec{U}}
\def\vV{\vec{V}}
\def\vW{\vec{W}}
\def\vX{\vec{X}}
\def\vY{\vec{Y}}
\def\vZ{\vec{Z}}

\def\valpha{\vec{\alpha}}
\def\vbeta{\vec{\beta}}
\def\vgamma{\vec{\gamma}}
\def\vdelta{\vec{\delta}}
\def\vepsilon{\vec{\epsilon}}
\def\veps{\vec{\varepsilon}}
\def\vzeta{\vec{\zeta}}
\def\veta{\vec{\eta}}
\def\vtheta{\vec{\theta}}
\def\viota{\vec{\iota}}
\def\vkappa{\vec{\kappa}}
\def\vlambda{\vec{\lambda}}
\def\vmu{\vec{\mu}}
\def\vnu{\vec{\nu}}
\def\vxi{\vec{\xi}}
\def\vpi{\vec{\pi}}
\def\vrho{\vec{\rho}}
\def\vsigma{\vec{\sigma}}
\def\vtau{\vec{\tau}}
\def\vupsilon{\vec{\upsilon}}
\def\vphi{\vec{\phi}}
\def\vchi{\vec{\chi}}
\def\vpsi{\vec{\psi}}
\def\vomega{\vec{\omega}}
\def\vvartheta{\vec{\vartheta}}
\def\vvarpi{\vec{\varpi}}
\def\vvarrho{\vec{\varrho}}
\def\vvarsigma{\vec{\varsigma}}
\def\vvarphi{\vec{\varphi}}

\def\vGamma{{\bf {\Gamma}}}
\def\vDelta{{\bf {\Delta}}}
\def\vTheta{{\bf {\Theta}}}
\def\vLambda{{\bf {\Lambda}}}
\def\vXi{{\bf {\Xi}}}
\def\vPi{{\bf {\Pi}}}
\def\vSigma{{\bf {\Sigma}}}
\def\vUpsilon{{\bf {\Upsilon}}}
\def\vPhi{{\bf {\Phi}}}
\def\vPsi{{\bf {\Psi}}}
\def\vOmega{{\bf {\Omega}}}

\def\mA{{\mathcal A}}
\def\mB{{\mathcal B}}
\def\mC{{\mathcal C}}
\def\mD{{\mathcal D}}
\def\mE{{\mathcal E}}
\def\mF{{\mathcal F}}
\def\mG{{\mathcal G}}
\def\mH{{\mathcal H}}
\def\mI{{\mathcal I}}
\def\mJ{{\mathcal J}}
\def\mK{{\mathcal K}}
\def\mL{{\mathcal L}}
\def\mM{{\mathcal M}}
\def\mN{{\mathcal N}}
\def\mO{{\mathcal O}}
\def\mP{{\mathcal P}}
\def\mQ{{\mathcal Q}}
\def\mR{{\mathcal R}}
\def\mS{{\mathcal S}}
\def\mT{{\mathcal T}}
\def\mU{{\mathcal U}}
\def\mV{{\mathcal V}}
\def\mW{{\mathcal W}}
\def\mX{{\mathcal X}}
\def\mY{{\mathcal Y}}
\def\mZ{{\mathcal Z}}

\def\vmA{\vec{\mathcal A}}
\def\vmB{\vec{\mathcal B}}
\def\vmC{\vec{\mathcal C}}
\def\vmD{\vec{\mathcal D}}
\def\vmE{\vec{\mathcal E}}
\def\vmF{\vec{\mathcal F}}
\def\vmG{\vec{\mathcal G}}
\def\vmH{\vec{\mathcal H}}
\def\vmI{\vec{\mathcal I}}
\def\vmJ{\vec{\mathcal J}}
\def\vmK{\vec{\mathcal K}}
\def\vmL{\vec{\mathcal L}}
\def\vmM{\vec{\mathcal M}}
\def\vmN{\vec{\mathcal N}}
\def\vmO{\vec{\mathcal O}}
\def\vmP{\vec{\mathcal P}}
\def\vmQ{\vec{\mathcal Q}}
\def\vmR{\vec{\mathcal R}}
\def\vmS{\vec{\mathcal S}}
\def\vmT{\vec{\mathcal T}}
\def\vmU{\vec{\mathcal U}}
\def\vmV{\vec{\mathcal V}}
\def\vmW{\vec{\mathcal W}}
\def\vmX{\vec{\mathcal X}}
\def\vmY{\vec{\mathcal Y}}
\def\vmZ{\vec{\mathcal Z}}

\def\tr{\mbox{\rm tr}}
\def\rank{\mbox{\rm rank}}
\def\diag{\mbox{\rm diag}}
\def\sign{\mbox{\rm sign}}

\def\MSE{\mbox{\rm MSE}}
\def\PMSE{\mbox{\rm PMSE}}
\def\RSS{\mbox{\rm RSS}}
\def\PRSS{\mbox{\rm PRSS}}
\def\Var{\mbox{\rm Var}}
\def\Cov{\mbox{\rm Cov}}
\def\Cor{\mbox{\rm Cor}}
\def\rmvec{\mbox{\rm vec}}
\def\Loss{\mbox{\rm Loss}}
\def\Diag{\mbox{\rm Diag}}

\def\thefootnote{\fnsymbol{footnote}}

\def\argmin{\arg\!\min}

\def\dsum{\displaystyle\sum}
\def\dprod{\displaystyle\prod}

\begin{center}
{\Large 
Plug-in Optimization Method for Penalty Parameters \\in Nonparametric GMANOVA model
}\\
(Last Modified: 2026.9.30.)
\\
 Isamu {\sc Nagai}\footnote{Corresponding author, E-mail: {\it inagai@lets.chukyo-u.ac.jp}}
\\[1.5mm]
\begin{small}
	{\it
	Faculty of Liberal Arts and Sciences, Chukyo University\\[-2mm]
	101-2 Yagoto Honmachi, Showa-ku, Nagoya, Aichi 466-8666 Japan
	}
\end{small}
\end{center}

\begin{abstract}
\baselineskip=17pt
In order to estimate the longitudinal trend, we often use the generalized multivariate analysis of variance (GMANOVA) model (Potthoff \& Roy, 1964), when the longitudinal data is balanced data.
Usually, we use this GMANOVA model with some polynomial at the time of measurement. 
However, when the longitudinal trend has flexible curve, we can not derive good fitting estimated curve when we use some polynomial curves.
Then, Nagai (2011) proposed the nonparametric GMANOVA model which uses on several known basis functions instead of using the polynomial curves.
If we use several basis functions, then overfitting problem is occurred.
Nagai (2011) also proposed the estimation method for avoiding overfitting and unstable problems, and reducing computational iterative algorithm by extending the generalized ridge regression model (Yanagihara, Nagai \& Satoh, 2009).
In the present paper, we extend one of the optimization methods in Nagai, Yanagihara and Satoh (2012) into the estimation method in Nagai (2011).
Through numerical studies, we show some properties of each optimization method.
\end{abstract}
\noindent 
{\it Key words}:
Generalized ridge regression;
GMANOVA model;
Longitudinal trend;
Non-iterative estimator;
Plug-in optimization method;
Repeat plug-in optimization method;
Varying coefficient model.

\vskip 8pt
\centerline{\bf \large 1. Introduction}
\setcounter{section}{1}
\setcounter{equation}{0}
\vskip 8pt

The longitudinal data are repeatedly taken on each individual over time.
In the present paper, we consider these data are taken at the same timing over all individuals.
This data is called a balanced data.
For analyzing this data, the generalized multivariate analysis of variance (GMANOVA) model (Potthoff \& Roy, 1964) is often used.
This model can be seen as a varying coefficient model with some assumption (Satoh \& Yanagihara, 2010).

The GMANOVA model is written as follows;
\begin{eqnarray}
\vY=\1_n\vmu'\vX'+\vA\vXi\vX'+\vmE,
\label{Eq13}
\end{eqnarray}
where $\vY=(\vy_1,\ldots,\vy_n)'$ is a known $n \times p$ matrix, $\vy_i$ is a longitudinal data of $i$th individual,
$\vA$ is a known $n \times k$ matrix whose $i$th row is organized from the $i$th individual's explanatory variables,
$\vX=(\vx(t_1),\ldots,\vx(t_p))'$ is a known $p \times q$ matrix which is made from some function of the $j$th measurement time $t_j$,
$\vmu$ is an unknown $q$-dimensional vector,
$\vXi$ is an unknown $p \times q$ matrix,
$\vmE=(\veps_1,\ldots,\veps_n)'$ is an $n \times p$ error matrix,
and $\1_n$ is $n$-dimensional one vector.
In the present paper, we assume $E[\veps_i]=\0_p$, $\Cov[\veps_i]=\vSigma$ ($i=1,\ldots,n$) which is an unknown positive definite $p \times p$ matrix, $\veps_i \indep \veps_j$ ($i\neq j$), $n-k-p-1>0$, $\rank(\vX)=q$, $\rank(\vA)=k$ and $\vA'\1_n=\0_k$ where $\0_r$ is a $r$-dimensional zero vector.
Here, $\vA'\1_n=\0_k$ means each column of $\vA$ is centralized, respectively.
Addition to this, we also assume $\rank((\vI_n-\1_n\1_n'/n-\vA(\vA'\vA)^{-1}\vA')\vY)=p$.

In the ordinary GMANOVA model, $\vx(t)$ is set as $(1,t,\ldots,t^{q-1})$.
Then, by using thise $\vx(t)$, we can estimate the longitudinal trend and varying coefficient curve by using the $(q-1)$-th degree polynomial at the time of measurement.
Since this ordinary method can not estimate the trend and curve when they have flexible curves, Nagai (2011) proposed the nonparametric estimation method by using several known basis functions instead of using the polynomial curves.
Their method is based on Riedel and Imre (1993) and Kshirsagar and Smith (1995).
At first, Nagai (2011) proposed the least squares estimators for $\vmu$ and $\vXi$ as $\hat{\vmu}=n^{-1}(\vX'\vX)^{-1}\vX'\vY'\1_n$ and $\hat{\vXi}=(\vA'\vA)^{-1}\vA'\vY\vX(\vX'\vX)^{-1}$, and we call them as LSEs.
When we use $\hat{\vmu}$ and $\hat{\vXi}$ based on the known basis function, there is overfitting problem.
Together with this, if there is high correlation columns in $\vA$, then $\hat{\vXi}$ becomes unstable.
In ordinary penalize idea, we use some penalized estimation method in order to avoid these problems.
However, this ordinary idea needs many iterative computational algorithm for obtaining estimator and optimizing penalty term.
Nagai (2011) also proposed the penalized method for avoiding overfitting and unstable problems in order to reduce computational algorithm.
Their method is derived by extending the generalized ridge regression (see, e.g., Yanagihara, Nagai and Satoh (2011), and Hoerl and Kennard (1970)).

In Nagai (2011), there is $C_p$-type information criterion for optimizing penalty parameters.
On the other hand, in the generalized ridge regression, there are another optimization method for the penalty parameters  is not only $C_p$ criterion but also other plug-in type optimization method (Nagai, Yanagihara \& Satoh, 2012).
In the present paper, we propose the plug-in optimization method by extending one of their optimization methods.
Furthermore, we propose the repeat plug-in optimization method.
Moreover, we prove the magnitude relationship between the repeat plug-in optimized parameter under some assumptions.

The remainder is organized as follows;
In Section 2, we illustrate the penalized estimator which is proposed by Nagai (2011).
Further, we also introduce the $C_p$-type criterion.
In Section 3, we propose a new optimization method based on the plug-in optimization method.
As well as this, we propose a repeat plug-in optimization method and derive some properties in them.
In Section 4, we compare these proposed estimation methods and LSEs.
We conclude this paper in Section 5.
Technical details  about some calculation are in Appendix.

\vskip 8pt
\centerline{\bf \large 2. Penalized Estimator and $C_p$-type criterion}
\setcounter{section}{2}
\setcounter{equation}{0}
\vskip 8pt

There is overfitting problem when we estimate the longitudinal trend or varying coefficient curves based on some known basis functions with using LSEs.
In addition, when there are high correlation columns in $\vA$, $\hat{\vXi}$ becomes unstable.
In order to avoid these problems, Nagai (2011) proposed the penalized estimator as follows;
\begin{eqnarray*}
\begin{tabular}{l}
$\hat{\vmu}_\lambda=(\vX'\vX+\lambda\vK)^{-1}\vX'\vY'\1_n/n$,\\[1mm]
$\hat{\vXi}_{\vtheta,\lambda}=(\vA'\vA+\vQ\vTheta\vQ')^{-1}\vA'\vY\vX(\vX'\vX+\lambda\vK)^{-1},$
\end{tabular}
\label{Eq21}
\end{eqnarray*}
where $\lambda$ is a nonnegative penalty parameter,
$\vTheta=\diag(\vtheta)$,
$\vtheta=(\theta_1,\ldots,\theta_k)'$,
$\theta_i$ is also a nonnegative penalty parameter,
$\vK$ is a known $q \times q$ positive definite penalty matrix,
and $\vQ$ is an orthogonal matrix which diagonalizes $\vA'\vA$, that is $\vQ'\vA'\vA\vQ=\vD$ with $\vQ'\vQ=\vQ\vQ'=\vI_k$ where $\vD=\diag(d_1,\ldots,d_k)$ and $d_i$ is the $i$th eigenvalue of $\vA'\vA$.
In these estimators, we have to optimize $(k+1)$-parameters in $\vtheta$ and $\lambda$.
Note that $d_i>0$ since we assume $\rank(\vA)=k$.

Using these estimators, we can avoid the overfitting problem when we use basis functions for $\vX$ by using $\lambda$ and $\vK$, and unstably problem by using $\vtheta$.
We note that these estimators coincide with the LSEs (Nagai, 2011) when $\vtheta=\0_k$ and $\lambda\vK=\0_q\0_q'$.
In order to optimize $(k+1)$ parameters in $\vtheta$ and $\lambda$, we consider following the predicted mean squared error (PMSE);
\begin{eqnarray}
\PMSE[\hat{\vY}_{\vtheta,\lambda}]=E_\vY\left[E_\vU\left[\tr\left\{\left(\vU-\hat{\vY}_{\vtheta,\lambda}\right)\vSigma^{-1}\left(\vU-\hat{\vY}_{\vtheta,\lambda}\right)'\right\}\right]\right],
\label{Eq22}
\end{eqnarray}
where $\hat{\vY}_{\vtheta,\lambda}=\1_n\hat{\vmu}_\lambda'\vX'+\vA\hat{\vXi}_{\vtheta,\lambda}\vX'$, $\vU$ is a random matrix which is independent distributed according to the same distribution of $\vY$, and $E_\vV[\cdot]$ means the expectation of the distribution of $\vV$.
Here, the PMSE evaluates the standardized distance between the predicted value $\hat{\vY}_{\vtheta,\lambda}$ and the new response value $\vU$.
We consider estimating the optimal $\vtheta$ and $\lambda$ which minimize $\PMSE[\hat{\vY}_{\vtheta,\lambda}]$.

Here, we note that we need some unknown matrices in order to calculate this PMSE.
As similar as Yanagihara, Nagai and Satoh (2009) and Nagai, Yanagihara and Satoh (2012), we consider several optimization methods two type methods.
One is minimizing the estimator for PMSE which is derived by Nagai (2011).
Another ones is minimizing the PMSE with some unknown values and insert some estimators into there.
In this paper, we propose the second method.
In here, we introduce the first method.

Nagai (2011) proposed the $C_p$-type criterion for estimating \eref{Eq22}.
Further, Nagai (2011) proposed the modified $C_p$ ($MC_p$) criterion with assumption $\vmE\sim N_{n\times p}(\0_n\0_p',\vSigma\otimes\vI_n)$ and $n-k-p-2>0$.
This $MC_p$ criterion is derived by correcting the bias between $C_p$ criterion and $\PMSE[\hat{\vY}_{\vtheta,\lambda}]$ with these assumptions.
(Original $C_p$ criterion for selecting variables in the model \eref{Eq13} with $\vX=\vI_p$ and $p=1$ was proposed by Mallows (1973; 1995), Sparks, Coutsourides and Troskie (1983) also proposed the criterion in the model \eref{Eq13} with $\vX=\vI_p$. Moreover, several authors proposed bias corrected $C_p$-type criteria in e.g., Yanagihara and Satoh (2010), and Fujikoshi and Satoh (1997).)
Neglecting several constant term, their $C_p$-type criteria are expressed as follows;
\begin{eqnarray*}
C_p(\vtheta,\lambda)=&\tr\{(\vY-\hat{\vY}_{\vtheta,\lambda})\vS^{-1}(\vY-\hat{\vY}_{\vtheta,\lambda})'\}+2\tr(\vG_\lambda)\{1+\tr(\vD(\vD+\vTheta)^{-1})\}\\
MC_p(\vtheta,\lambda)=&c_{\rm{\tiny M}}\tr\{(\vY-\hat{\vY}_{\vtheta,\lambda})\vS^{-1}(\vY-\hat{\vY}_{\vtheta,\lambda})'\}+2\tr(\vG_\lambda)\{1+\tr(\vD(\vD+\vTheta)^{-1})\},
\end{eqnarray*}
where $\vG_\lambda=\vX(\vX'\vX+\lambda\vK)^{-1}\vX'$, $\vS=\vY'\{\vI_n-\1_n\1_n'/n-\vA(\vA'\vA)^{-1}\vA'\}\vY/(n-k-1)$ and $c_{\rm{\tiny M}}=1-(p+1)/(n-k-1)$.
Note that $\vG_\lambda$ is a symmetric matrix.
Nagai (2011) also showed that the optimized $\vtheta$ by minimizing these criteria is derived by closed form when $\lambda$ is given.
Moreover, the optimization methods for $\lambda$ and $q$ are also proposed in there.
More details are showed in Nagai (2011).
\if01
On another front, there are another optimization methods in the generalized ridge regression (see, e.g., Nagai, Yanagihara and Satoh (2012)).
These optimization methods do not depend on these criteria.
One of these methods is considered as directly minimizing $\PMSE[\hat{\vY}_{\vtheta,\lambda}]$ and plug-in several LSEs into unknown variables.
In addition to this, these methods use the plug-inned estimator instead of using LSEs.
In this paper, we extend these plug-in optimization methods for our estimators in next subsection.
\fi

\vskip 8pt
\centerline{\bf \large 3. New optimization method based on the Minimizing PMSE}
\setcounter{section}{3}
\setcounter{equation}{0}
\vskip 6pt
\centerline{\bf 3.1. Plug-in optimization}
\vskip 6pt

Since these criteria estimate $\PMSE[\hat{\vY}_{\vtheta,\lambda}]$ in \eref{Eq22}, minimizing each criterion does not always correspond to minimizing $\PMSE[\hat{\vY}_{\vtheta,\lambda}]$.
Here, we consider directly minimizing $\PMSE[\hat{\vY}_{\vtheta,\lambda}]$.

Firstly, since $E_\vY[\vY]=E_\vU[\vU]$, we can see that $$\PMSE[\hat{\vY}_{\vtheta,\lambda}]=E_\vY\left[\tr\left\{\left(\hat{\vY}_{\vtheta,\lambda}-E_\vY[\vY]\right)\vSigma^{-1}\left(\hat{\vY}_{\vtheta,\lambda}-E_\vY[\vY]\right)'\right\}\right]+np.$$
Thus, we consider minimizing the first term since $\vtheta$ and $\lambda$ depend on only it.
We can rewrite the purpose term as follows (in Appendix A.1, we show this equation);
\begin{eqnarray*}
f(\vtheta,\lambda|\vSigma,\vGamma,\vmu)&\stackrel{\rm def.}{=}&\tr\left(\vSigma^{1/2}\vG_\lambda\vSigma^{-1}\vG_\lambda\vSigma^{1/2}\right)+n\vmu'\vX'(\vG_\lambda-\vI_p)\vSigma^{-1}(\vG_\lambda-\vI_p)\vX\vmu\\
&& +\sum_{i=1}^k \varphi_i(\theta_i,\lambda|\vSigma,\vgamma_i)+\tr\left(\vGamma\vX'\vSigma^{-1}\vX\vGamma'\vD\right),
\end{eqnarray*}
since $\vG_\lambda=\vG_\lambda'$, where $\beta_{1,i}(\lambda|\vSigma,\vgamma_i)=d_i\vgamma_i'\vX'\vG_\lambda\vSigma^{-1}\vX\vgamma_i$, $\beta_{2,i}(\lambda|\vSigma,\vgamma_i)=d_i\vgamma_i'\vX'\vG_\lambda\vSigma^{-1}\vG_\lambda'\vX\vgamma_i$, $\alpha_{i}(\lambda|\vSigma,\vgamma_i)=\tr(\vSigma\vG_\lambda$ $\vSigma^{-1}\vG_\lambda')+\beta_{2,i}(\lambda|\vSigma,\vgamma_i)$, $\vGamma=(\vgamma_1,\ldots,\vgamma_k)'=\vQ'\vXi$, and 
\begin{eqnarray}
\varphi_i(\theta_i,\lambda|\vSigma,\vgamma_i)=\left(\dfrac{d_i}{d_i+\theta_i}\right)^2\alpha_{i}(\lambda|\vSigma,\vgamma_i)-\dfrac{2d_i}{d_i+\theta_i}\beta_{1,i}(\lambda|\vSigma,\vgamma_i).
\label{Eq23}
\end{eqnarray}
Thus, minimizing PMSE in \eref{Eq22} coincides with minimizing $\varphi_i(\theta_i,\lambda|\vSigma,\vgamma_i)$ in the above for each $\theta_i$ ($i=1,\ldots,k$) when $\lambda$ is given.
Here, we note that $\alpha_i(\lambda|\vSigma,\vgamma_i)> 0$ from $\tr(\vSigma\vG_\lambda\vSigma^{-1}\vG_\lambda)> 0$ and $\beta_{2,i}(\lambda|\vSigma,\vgamma_i)> 0$ since $\vSigma$ is a positive definite matrix.
Further, even if we use another $p \times p$ positive definite $\vB$ matrix for $\vSigma$, $\beta_{2,i}(\lambda|\vB,\vr)>0$ for any suitable dimension non-zero vector $\vr$.
Then, the optimized $\theta_i$ by minimizing $\varphi_i(\theta_i,\lambda|\vSigma,\vgamma_i)$ for fixed $\lambda$, i.e., $\theta_i^*(\lambda|\vSigma,\vgamma_i)=\arg\!\min_{\theta_i\ge 0}\varphi(\theta_i,\lambda|\vSigma,\vgamma_i)$ is derived as follows (prove of this optimization is in Appendix A.2);
\begin{eqnarray}
\theta_i^*(\lambda|\vSigma,\vgamma_i)
=
\left\{{
\begin{tabular}{ll}
$0$ & (if $\beta_{1,i}(\lambda|\vSigma,\vgamma_i)>\alpha_{i}(\lambda|\vSigma,\vgamma_i)$) \\[1.5mm]
$\dfrac{d_i\{\alpha_{i}(\lambda|\vSigma,\vgamma_i)\!-\!\beta_{1,i}(\lambda|\vSigma,\vgamma_i)\}}{\beta_{1,i}(\lambda|\vSigma,\vgamma_i)}$ & (if $\alpha_i(\lambda|\vSigma,\vgamma_i)\ge\beta_{1,i}(\lambda|\vSigma,\vgamma_i)>0$)\\[1.5mm]
$\infty$ & (otherwise which means $0 \ge \beta_{1,i}(\lambda|\vSigma,\vgamma_i)$)\\[1.5mm]
\end{tabular}}\right..
\label{Eq24}
\label{PI}
\end{eqnarray}
This fact means that there are better estimator for $\vGamma$ or $\vXi$ than the ordinary estimator $\hat{\vXi}=\hat{\vXi}_{\0_k,0}$ based on the PMSE if $\alpha_i(\lambda|\vSigma,\vgamma_i) \ge \beta_{1,i}(\lambda|\vSigma,\vgamma_i)$ is satisfied for some $i$.

Since $\theta_i^*(\lambda|\vSigma,\vgamma_i)$ needs unknown matrices $\vSigma$ and $\vgamma_i$ which is the $i$th row of $\vGamma=\vQ'\vXi$, we substitute $\vS$ and $\hat{\vGamma}=\vQ'\hat{\vXi}$ into them in \eref{Eq24}.
Thus, we estimate $\theta_i^*(\lambda|\vSigma,\vgamma_i)$ as $\hat{\theta}_i(\lambda|\vS,\hat{\vgamma}_i)$ ($i=1,\ldots,k$) where $\hat{\vgamma}_i$ is the $i$th row of $\hat{\vGamma}$.
We denote $\hat{\vtheta}(\lambda|\vS,\hat{\vGamma})=(\hat{\theta}_1(\lambda|\vS,\hat{\vgamma}_1),\ldots,\hat{\theta}_k(\lambda|\vS,\hat{\vgamma}_k))'$.
We do not any iterative computational algorithm for optimizing $\vtheta$ since we can derive $\hat{\vtheta}(\lambda|\vS,\hat{\vGamma})$ in closed form.

The estimator for $\lambda$ is derived by minimizing $C_p(\hat{\vtheta}(\lambda|\vS,\hat{\vGamma}),\lambda)$ or $MC_p(\hat{\vtheta}(\lambda|\vS,\hat{\vGamma}),\lambda)$ based on some $q$.
From this reason, we can derive the optimized penalized parameters frome only one parameter minimization method for fixed $q$.
\if01
Letting $\hat{\lambda}=\arg\!\min_{\lambda \ge 0}f(\hat{\vtheta}(\lambda|\vS,\hat{\vGamma}),\lambda|\vS,\hat{\vGamma}_{\hat{\vtheta}(\lambda|\vS,\hat{\vGamma}),\lambda},\hat{\vmu}_\lambda)$ for fixed $q$, we can select $q$ which is the number of basis functions as $\arg\!\min_{q}f(\hat{\vtheta}(\hat{\lambda}|\vS,\hat{\vGamma}),\hat{\lambda}|\vS,\hat{\vXi}_{\hat{\vtheta}(\hat{\lambda}|\vS,\hat{\vGamma}),\hat{\lambda}},\hat{\vmu}_{\hat{\lambda}})$.
\fi

\vskip 6pt
\centerline{\bf 3.2. Repeat plug-in optimization}
\vskip 6pt

Nagai, Yanagihara and Satoh (2012) also proposed repetition plug-in optimization in order to avoid unstable in their model.
We extend their method into our model and optimization in \eref{Eq24}.
Here, $\hat{\vGamma}_{\vtheta,\lambda}=(\vgamma_1(\vtheta,\lambda),\ldots,\vgamma_k(\vtheta,\lambda))'=\vQ'\hat{\vXi}_{\vtheta,\lambda}=(\vD+\vTheta)^{-1}\vC'\vP\vY\vX(\vX'\vX+\lambda\vK)^{-1}$ where $\vP$ is an $n \times n$ orthogonal matrix which satisfies $\vP\vA\vQ=\vC=(\vD^{1/2},\0_k\0_{n-k}')'$, and note that $\hat{\vGamma}=\hat{\vGamma}_{\0_k,0}$.
Then, we can rewrite $\hat{\vgamma}_i(\vtheta,\lambda)$ as $\hat{\vgamma}_i(\theta_i,\lambda)=\sqrt{d_i}(\vX'\vX+\lambda\vK)^{-1}\vX'\vY'\vP'\ve_i/(d_i+\theta_i)$ where $\ve_i$ is an $n$-dimensional vector with the $i$th element is $1$ and other elements are zeros.
From the above expression, we note $\hat{\vgamma}_i(\vtheta,\lambda)$ is only depending on $\theta_i$ in $\vtheta$.
Add to this, when $0\ge \beta_{1,i}(\lambda|\vSigma,\vgamma_i)$ is satisfied for some $i$, we derive better estimator without using the $i$th eigenvalue $d_i$.

When we use $\hat{\vGamma}$ and $\vS$, we can estimate $\theta_i^*(\lambda|\vSigma,\vgamma_i)$ ($i=1,\ldots,k$) in \eref{PI}.
However, since $\hat{\vGamma}$ also has same problems with $\hat{\vXi}$ in sometimes, the $\hat{\vtheta}(\lambda|\vS,\vGamma)$ is also unstable in then.
Hence, in order to aboid this unstable issue, we use $\hat{\vGamma}_{\hat{\vtheta}(\lambda|\vS,\hat{\vGamma})}$ instead of using $\hat{\vGamma}$.
Then, we renew $\hat{\vtheta}(\lambda|\vS,\hat{\vGamma})$ as $\hat{\vtheta}(\lambda|\vS,\hat{\vGamma}_{\hat{\vtheta}(\lambda|\vS,\hat{\vGamma})})$.
We note that the $i$th element of $\hat{\vtheta}(\lambda|\vS,\hat{\vGamma}_{\hat{\vtheta}(\lambda|\vS,\vGamma)})$ is only depend on $\hat{\vgamma}_i(\hat{\theta}_i(\lambda|\vS,\hat{\gamma}_i),\lambda)$.
Hence, we rewrite the renewaled parameter $\hat{\vtheta}(\lambda|\vS,\hat{\vGamma}_{\hat{\vtheta}(\lambda|\vS,\hat{\vGamma})})$ as $\hat{\vtheta}(\lambda|\vS,\tilde{\vtheta})$ where $\tilde{\vtheta}$ means the previous parameter.
For example, the renewal parameter $\vtheta$ with new parameter $\lambda$, namely $\hat{\vtheta}^{[2]}(\lambda|\vS,\tilde{\vtheta}^{[1]})$ where $\tilde{\vtheta}^{[1]}=(\hat{\theta}_1(\lambda|\vS,\hat{\vgamma}_1),\ldots,\hat{\theta}_k(\lambda|\vS,\hat{\vgamma}_k))'$. 
Furthermore, we can repeat this method as using the estimator for $\vGamma$ based on the previous estimated parameter $\vtheta$ and renewal estimated parameter.
Hereafter, we denote the previous estimated $\vtheta$, which is used for deriving $(s+1)$-step $\vtheta$, as $\tilde{\vtheta}^{[s]}=(\tilde{\theta}_1^{[s]},\ldots,\tilde{\theta}_k^{[s]})'$ for some natural number $s$, and $\tilde{\vtheta}^{[0]}=\hat{\vtheta}^{[0]}=\0_k$ for simple expression.
In summary, we repeat the optimization method as $\hat{\vtheta}^{[s+1]}(\lambda|\vS,\tilde{\vtheta}^{[s]})=(\hat{\theta}_1^{[s+1]}(\lambda|\vS,\tilde{\theta}_1^{[s]}),\ldots,\hat{\theta}_k^{[s+1]}(\lambda|\vS,\tilde{\theta}_k^{[s]}))'$ since the expression for $\hat{\vgamma}_i(\vtheta,\lambda)$ can be rewritten as $\hat{\vgamma}_i(\theta_i,\lambda)$. 

Further, in order to obtain $\hat{\theta}_i^{[s]}(\lambda|\vS,\tilde{\theta}_i^{[s-1]})$, we use the estimators for $\beta_{1,i}(\lambda|\vSigma,\vgamma_i)$,  $\beta_{2,i}(\lambda|\vSigma,\vgamma_i)$ and  $\alpha_{i}(\lambda|\vSigma,\vgamma_i)$ based on using $\tilde{\theta}_i^{[s-1]}$, for each natural number $s$ and $i=1,\ldots,k$.
In similar to the expression of $\hat{\theta}_i^{(s+1)}(\lambda|\vS,\tilde{\theta}_i^{(s)})$, we write  $\hat{\beta}_{1,i}(\lambda|\vS,\tilde{\theta}_i^{[s-1]})$, $\hat{\beta}_{2,i}(\lambda|\vS,\tilde{\theta}_i^{[s-1]})$ and $\hat{\alpha}_i(\lambda|\vS,\tilde{\theta}_i^{[s-1]})$ for the each estimator for $\beta_{1,i}(\lambda|\vSigma,\vgamma_i)$, $\beta_{2,i}(\lambda|\vSigma,\vgamma_i)$ and $\alpha_{i}(\lambda|\vSigma,\vgamma_i)$, respectively.
More detail expressions for these estimators are derived as $\hat{\beta}_{1,i}(\lambda|\vS,\theta_i)=d_i^2\ve_i'\vP\vY\vG_\lambda^2\vS^{-1}\vG_\lambda\vY'\vP'\ve_i/(d_i+\theta_i)^2$ and $\hat{\beta}_{2,i}(\lambda|\vS,\theta_i)=d_i^2\ve_i'\vP\vY\vG_\lambda^2\vS^{-1}\vG_\lambda^2\vY'\vP'\ve_i/(d_i+\theta_i)^2$ by substituting $\hat{\theta}_i^{[s]}(\lambda|\vS,\tilde{\vtheta}^{[s-1]})$ to $\theta_i$ for any natural number $s \ge 2$.
In spite of this, since we use $\hat{\vGamma}$ for $s=1$, these estimators are derived as $\hat{\beta}_{1,i}(\lambda|\vS,\theta_i)=\ve_i'\vP\vY\vG_\lambda\vS^{-1}\vG_0\vY'\vP'\ve_i$ and $\hat{\beta}_{2,i}(\lambda|\vS,\theta_i)=\ve_i'\vP\vY\vG_\lambda\vS^{-1}\vG_\lambda\vY'\vP'\ve_i$ when $s=1$.
Moreover, the estimator for $\alpha_i$ is derived as $\hat{\alpha}_i(\lambda|\vS,\theta_i)=\tr(\vS\vG_\lambda\vS^{-1}\vG_\lambda)+\hat{\beta}_{2,i}(\lambda|\vS,\theta_i)$.

Using these expressions, we obtain the magnitude relations among the $i$th element $\hat{\theta}_i^{[s]}(\lambda|\vS,\tilde{\theta}_i^{[s-1]})$ when $\lambda$ is set as same value for any natural number $s$.
Firstly, if we do not need $\theta_i$ when $d_i$ is too large, we wanna set $\hat{\theta}_i^{[s]}(\lambda|\vS,\tilde{\theta}_i^{[s-1]})$ as $0$ since there corresponding parameter in $\vGamma$ or $\vXi$ is stable.
This is achieved by the following theorem;
\begin{theorem}
$\hat{\theta}_i^{[s+1]}(\lambda|\vS,\tilde{\theta}_i^{[s]})=0$ when $\hat{\theta}_i^{[s]}(\lambda|\vS,\tilde{\theta}_i^{[s-1]})=0$ for some natural number $s \ge 2$.
\label{t1}
\end{theorem}
\begin{proof}
When $\hat{\theta}_i^{[s]}(\lambda|\vS,\tilde{\theta}_i^{[s-1]})=0$, the condition $\hat{\beta}_{1,i}(\lambda|\vS,\tilde{\theta}_i^{[s-1]})>\hat{\alpha}_{1,i}(\lambda|\vS,\tilde{\theta}_i^{[s-1]})$ is satisfied.
In $s\ge 2$, we can prove this theorem as some elementary calculation as following ways.
This inequation means that $d_i^2\ve_i'\{\vP\vY\vG_\lambda^2\vS^{-1}\vG_\lambda(\vI_p-\vG_\lambda)\}\vY'\vP'\ve_i/\{d_i+\hat{\theta}_i^{[s-1]}(\lambda|\vS,\tilde{\theta}_i^{[s-2]})\}^2 \ge \tr(\vS\vG_\lambda\vS^{-1}\vG_\lambda)$.
Since $\hat{\theta}_i^{[s-1]}(\lambda|\vS,\tilde{\theta}_i^{[s-2]}) \ge 0$ and $\tr(\vS\vG_\lambda\vS^{-1}\vG_\lambda)>0$, we obtain $\ve_i'\{\vP\vY\vG_\lambda^2\vS^{-1}\vG_\lambda(\vI_p-\vG_\lambda)\}\vY'\vP'\ve_i\ge d_i\ve_i'\{\vP\vY\vG_\lambda^2\vS^{-1}\vG_\lambda(\vI_p-\vG_\lambda)\}\vY'\vP'\ve_i/\{d_i+\hat{\theta}_i^{[s-1]}(\lambda|\vS,\tilde{\theta}_i^{[s-2]})\}^2$.
From these inequations, we can see that $\ve_i'\{\vP\vY\vG_\lambda^2\vS^{-1}\vG_\lambda(\vI_p-\vG_\lambda)\}\vY'\vP'\ve_i>\tr(\vS\vG_\lambda\vS^{-1}\vG_\lambda)$ is satisfied.
Thus, $\hat{\theta}_i^{[s+1]}(\lambda|\vS,\tilde{\theta}_i^{[s]})=0$ is derived for $s\ge 2$.
\end{proof}
This means the following parameter is derived as $0$ after the parameter becomes $0$ in some number of repeats.
In addition, when we obtain the estimated parameter as $0$, we do not need any repetition.
Notice that we do not say $\hat{\theta}_i^{[1]}(\lambda|\vS,\tilde{\theta}_i^{[0]})$ is always becoming $0$.

Secondary, when $d_i$ is too small, there are no information to estimate $\vGamma$ or $\vXi$.
Thus, we wanna set $\hat{\theta}_i^{[s+1]}(\lambda|\vS,\tilde{\theta}_i^{[s]})$ as $\infty$.
We achieve this by the following theorem;
\begin{theorem}
$\hat{\theta}_i^{[s+1]}(\lambda|\vS,\tilde{\theta}_i^{[s]})=\infty$ when $\hat{\theta}_i^{[s]}(\lambda|\vS,\tilde{\theta}_i^{[s-1]})=\infty$ for some natural number $s$.
\label{t2}
\end{theorem}
\begin{proof}
When $\hat{\theta}_i^{[s]}(\lambda|\vS,\tilde{\theta}_i^{[s-1]})=\infty$, we obtain $\hat{\beta}_{1,i}(\lambda|\vS,\tilde{\theta}_i^{[s]})=0$.
Then, from the estimation method for \eref{PI}, $\hat{\theta}_i^{[s+1]}(\lambda|\vS,\tilde{\theta}_i^{[s]})=\infty$.
\end{proof}
This theorem means that when this method judged the $i$th eigenvalue is too small for some number of repeats, that is there is no information, we do not use the $i$th eigenvalue's information after then.
Moreover, this theorem also means that we do not need any repetition when we obtain the estimated parameter as $\infty$.

Lastly, if $\hat{\theta}_i^{[s]}(\lambda|\vS,\tilde{\theta}_i^{[s-1]}) \in (0,\infty)$, we consider the relationship between $\hat{\theta}_i^{[s]}(\lambda|\vS,\tilde{\theta}_i^{[s-1]})$ and $\hat{\theta}_i^{[s+1]}(\lambda|\vS,\tilde{\theta}_i^{[s]})$.
Then, we obtain the following theorem;
\begin{theorem}
If $\hat{\theta}_i^{[s]}(\lambda|\vS,\tilde{\theta}_i^{[s-1]}) \in (0,\infty)$ for any natural number $s \ge 2$ and $\hat{\theta}_i^{[2]}(\lambda|\vS,\tilde{\theta}_i^{[1]})>\hat{\theta}_i^{[1]}(\lambda|\vS,\tilde{\theta}_i^{[0]})$, then
$\hat{\theta}_i^{[s+1]}(\lambda|\vS,\tilde{\theta}_i^{[s]}) > \hat{\theta}_i^{[s]}(\lambda|\vS,\tilde{\theta}_i^{[s-1]})$ for all natural number $s$.
\label{t3}
\end{theorem}
\begin{proof}
When $\hat{\theta}_i^{[s]}(\lambda|\vS,\tilde{\theta}_i^{[s-1]})<\infty$, this assumption means $\ve_i'\vP\vY\vG_\lambda^2\vS^{-1}\vG_\lambda\vY'\vP'\ve_i>0$.
Then, it holds that $\hat{\theta}_i^{[s+1]}(\lambda|\vS,\tilde{\theta}_i^{[s]})<\infty$.

Firstly, when $\hat{\theta}_i^{[2]}(\lambda|\vS,\tilde{\theta}_i^{[1]}) \in (0,\infty)$, $\hat{\alpha}(\lambda|\vS,\tilde{\theta}_i^{[1]})>\hat{\beta}_{i,1}(\lambda|\vS,\tilde{\theta}_i^{[1]})>0$.
Then, since $(d_i/(d_i+\hat{\theta}_i^{[1]}(\lambda|\vS,\tilde{\theta}_i^{[0]})))^2>0$ and $\tr(\vS\vG_\lambda\vS^{-1}\vG_\lambda)>0$, by using $\{1+\hat{\theta}_i^{[1]}(\lambda|\vS,\tilde{\theta}_i^{[0]})/d_i\}^2>1$ and $\hat{\theta}_i^{[2]}(\lambda|\vS,\tilde{\theta}_i^{[1]})>\hat{\theta}_i^{[1]}(\lambda|\vS,\tilde{\theta}_i^{[0]})$, we obtain that $\hat{\theta}_i^{[3]}(\lambda|\vS,\tilde{\theta}_i^{[2]}) \in (0,\infty)$ is also satisfied when $\hat{\theta}_i^{[2]}(\lambda|\vS,\tilde{\theta}_i^{[1]}) \in (0,\infty)$.
By computing $\hat{\theta}_i^{[3]}(\lambda|\vS,\tilde{\theta}_i^{[2]})-\hat{\theta}_i^{[2]}(\lambda|\vS,\tilde{\theta}_i^{[1]})$, we obtain $\hat{\theta}_i^{[3]}(\lambda|\vS,\tilde{\theta}_i^{[2]}) > \hat{\theta}_i^{[2]}(\lambda|\vS,\tilde{\theta}_i^{[1]})$ if $\hat{\theta}_i^{[2]}(\lambda|\vS,\tilde{\theta}_i^{[1]}) > \hat{\theta}_i^{[1]}(\lambda|\vS,\tilde{\theta}_i^{[0]})$ since $\tr(\vS\vG_\lambda\vS^{-1}\vG_\lambda) > 0$, $\hat{\theta}_i^{[2]}(\lambda|\vS,\tilde{\theta}_i^{[1]})>0$, $d_i>0$, and $e_i'\vP\vY\vG_\lambda^2\vS^{-1}\vG_\lambda\vY'\vP'\ve_i>0$ from the condition for $\hat{\theta}_i^{[2]}(\lambda|\vS,\tilde{\theta}_i^{[1]})\in (0,\infty)$.

Secondary, we assume $\hat{\theta}_i^{[m+1]}(\lambda|\vS,\tilde{\theta}_i^{[m]}) > \hat{\theta}_i^{[m]}(\lambda|\vS,\tilde{\theta}_i^{[m-1]})$ and $\hat{\theta}_i^{[m]}(\lambda|\vS,\tilde{\theta}_i^{[m-1]}) \in (0,\infty)$ for some natural number $m$.
Since we assume $\hat{\theta}_i^{[m]}(\lambda|\vS,\tilde{\theta}_i^{[m-1]}) \in (0,\infty)$, the following inequation is holding;
\begin{align*}
\left(\!1+\dfrac{\hat{\theta}_i^{[m-1]}(\lambda|\vS,\tilde{\theta}_i^{[m-2]})}{d_i}\right)^2\!\tr(\vS\vG_\lambda\vS^{-1}\vG_\lambda)+\ve_i'\vP\vY\vG_\lambda^2\vS^{-1}\vG_\lambda^2\vY'\vP'\ve_i>\ve_i'\vP\vY\vG_\lambda^2\vS^{-1}\vG_\lambda\vY'\vP'\ve_i>0.
\end{align*}
Using the assumption $\hat{\theta}_i^{[m+1]}(\lambda|\vS,\tilde{\theta}_i^{[m]}) > \hat{\theta}_i^{[m]}(\lambda|\vS,\tilde{\theta}_i^{[m-1]})$, we obtain $\hat{\theta}_i^{[m+2]}(\lambda|\vS,\tilde{\theta}_i^{[m+1]}) \in (0,\infty)$.
Then, we derive 
\begin{align*}
&\hat{\theta}_i^{[m+2]}(\lambda|\vS,\tilde{\theta}_i^{[m+1]})-\hat{\theta}_i^{[m+1]}(\lambda|\vS,\tilde{\theta}_i^{[m]})\\
=&\dfrac{\{2d_i+\hat{\theta}_i^{[m+1]}(\lambda|\vS,\tilde{\theta}_i^{[m]}) + \hat{\theta}_i^{[m]}(\lambda|\vS,\tilde{\theta}_i^{[m-1]})\}\{\hat{\theta}_i^{[m+1]}(\lambda|\vS,\tilde{\theta}_i^{[m]}) - \hat{\theta}_i^{[m]}(\lambda|\vS,\tilde{\theta}_i^{[m-1]})\}\tr(\vS\vG_\lambda\vS^{-1}\vG_\lambda)}{d_i\ve_i'\vP\vY\vG_\lambda^2\vS^{-1}\vG_\lambda\vY'\vP'\ve_i}.
\end{align*}
Then, we obtain $\hat{\theta}_i^{[m+2]}(\lambda|\vS,\tilde{\theta}_i^{[m+1]}) > \hat{\theta}_i^{[m+1]}(\lambda|\vS,\tilde{\theta}_i^{[m]})$ from the assumption $\hat{\theta}_i^{[m+1]}(\lambda|\vS,\tilde{\theta}_i^{[m]}) > \hat{\theta}_i^{[m]}(\lambda|\vS,\tilde{\theta}_i^{[m-1]})$, $\tr(\vS\vG_\lambda\vS^{-1}\vG_\lambda) > 0$, $d_i>0$, $\hat{\theta}_i^{[m+1]}(\lambda|\vS,\tilde{\theta}_i^{[m]})>0$, and $\hat{\theta}_i^{[m]}(\lambda|\vS,\tilde{\theta}_i^{[m-1]})>0$.

Using mathematical induction, we prove this theorem.
\end{proof}
From this proof, and using \tref{t1} and \tref{t2}, we note the following corollary;
\begin{corollary}
$\hat{\theta}_i^{[s+1]}(\lambda|\vS,\tilde{\theta}_i^{[s]}) \in (0,\infty)$ is holding only when $\hat{\theta}_i^{[s]}(\lambda|\vS,\tilde{\theta}_i^{[s-1]}) \in (0,\infty)$ for $s \ge 2$.
\end{corollary}
When $\hat{\theta}_i^{[1]}(\lambda|\vS,\tilde{\theta}_i^{[0]}) \in (0,\infty)$, we can see that $\hat{\theta}_i^{[1]}(\lambda|\vS,\tilde{\theta}_i^{[0]})\ge \hat{\theta}_i^{[0]}$ since $\hat{\theta}_i^{[0]}=\tilde{\theta}_i^{[0]}=0$.
That is to say, $\hat{\theta}_i^{[s]}(\lambda|\vS,\tilde{\theta}_i^{[s-1]})$ is a monotone increase sequence with $s$ when $\hat{\theta}_i^{[1]}(\lambda|\vS,\tilde{\theta}_i^{[0]}) \in (0,\infty)$ and $\hat{\theta}_i^{[2]}(\lambda|\vS,\tilde{\theta}_i^{[1]})>\hat{\theta}_i^{[1]}(\lambda|\vS,\tilde{\theta}_i^{[0]})$.
Hence, when we use large $s$, more shrinked estimator for $\vXi$ or $\vGamma$ is derived.
%
Further, the corresponding theorem when $\lambda\vK=\0_q\0_q'$ and $\vX=\vI_p$ is showed in Nagai, Yanagihara and Satoh (2012).

Similar to Nagai, Yanagihar and Satoh (2012), we can consider a convergence of $\hat{\theta}_i^{[s]}(\lambda|\vS,\tilde{\theta}_i^{[s-1]}) \in (0,\infty)$ when $s \to \infty$.
Then, we prove the following theorem when $\hat{\theta}_i^{[2]}(\lambda|\vS,\tilde{\theta}_i^{[1]})>\hat{\theta}_i^{[1]}(\lambda|\vS,\tilde{\theta}_i^{[0]})$;
\begin{theorem}
Letting $\nu_i=d_i/\{(d_i+\hat{\theta}_i^{[s]}(\lambda|\vS,\tilde{\theta}_i^{[s-1]})\}$, $L_i=\nu_i^2\tr(\vS\vG_\lambda\vS^{-1}\vG_\lambda)/\hat{\beta}_{1,i}(\lambda|\vS,\tilde{\theta}_i^{[s-1]})$ and $M_i=1-\hat{\beta}_{2,i}(\lambda|\vS,\tilde{\theta}_i^{[s-1]})/\hat{\beta}_{1,i}(\lambda|\vS,\tilde{\theta}_i^{[s-1]})$, if $\hat{\theta}_i^{[2]}(\lambda|\vS,\tilde{\theta}_i^{[1]})>\hat{\theta}_i^{[1]}(\lambda|\vS,\tilde{\theta}_i^{[0]})$, then we obtain $\hat{\theta}_i^{[\infty]}(\lambda|\vS)\stackrel{\rm def.}{=}\lim_{s\to \infty}\hat{\theta}_i^{[s]}(\lambda|\vS,\tilde{\theta}_i^{[s-1]}) \in (0,\infty)$ $(i=1,\ldots,k)$ as follwos;
$$\hat{\theta}_i^{[\infty]}(\lambda|\vS)=\left\{\begin{tabular}{ll} $\dfrac{d_i\{1-2L_i-\sqrt{1-4L_i(1-M_i)}\}}{2L_i}$ & {\rm ($4L_i(1-M_i) \le 1$)} \\ $\infty$ & {\rm ($4L_i(1-M_i) > 1$)}  \end{tabular}\right..$$
\label{t4}
\end{theorem}
\begin{proof}
Using $L_i$ and $M_i$, we can rewrite $\hat{\theta}_i^{[s+1]}(\lambda|\vS,\tilde{\theta}_i^{[s]})$ as $\eta_i^{[s+1]}=(1+\eta_i^{[s]})^2L_i-M_i$ where $\eta_i^{[s]}=\nu_i^{-1}-1=\hat{\theta}_i^{[s]}(\lambda|\vS,\tilde{\theta}_i^{[s-1]})/d_i$ for natural number $s$ and $\eta_i^{[0]}=0$.
From \tref{t3}, $\eta_i^{[s]}$ is also monotone increasing sequence in this situation.
If $\eta_i^{[s]}<1/(2L_i)-1$, we note that $\eta_i^{[s+1]}<1/(2L_i)-1$ when $4L_i(1-M_i)\le 1$.
Thus, when $4L_i(1-M_i)\le 1$, the sequence of $\eta_i^{[s]}$ is convergence in $s \to \infty$.
In this case, the sequence of $\hat{\theta}_i^{[s]}(\lambda|\vS,\tilde{\theta}_i^{[s-1]})$ is also convergence.
There is exist the converged value $\eta_i^*$ which satisfies $\eta_i^*=(1+\eta_i^*)^2L_i-M_i$.
When $4L_i(1-M_i)\le 1$, we obtain $\eta_i^*=(1\pm\sqrt{1-4L_i(1-M_i)})/(2L_i)-1$.
Let $\eta_{1,i}^*=(1-\sqrt{1-4L_i(1-M_i)})/(2L_i)-1$ and $\eta_{2,i}^*=(1+\sqrt{1-4L_i(1-M_i)})/(2L_i)-1$.
We note that $\eta_{1,i}$ and $\eta_{2,i}$ are positive constant.
Letting $h(x)=(1+x)^2L_i-M_i$, we obtain $h'(x)=d h(x)/(dx)=2(1+x)L_i$.
We obtain $|h'(\eta_{1,i})|<1$ and $|h'(\eta_{2,i})|>1$.
Thus, the stable convergence value is $\eta_{1,i}$ from the Proposition 1.9 in Galor (2007).
On the other hand, when $4L_i(1-M_i)> 1$, $\eta_i^{[s]} \to \infty$ ($s\to \infty$), and then $\hat{\theta}_i^{[\infty]}(\lambda|\vS)=\infty$, since these sequences are monotone increasing.
Hence, we prove this convergence.
\end{proof}
This theorem means that when the condition $4L_i(1-M_i)>1$ is satisfied, we remove the effect of the $i$th eigenvalue $d_i$ if we use this converged parameter.
Furthermore, we can find that this covergence is coinsides with that of Nagai, Yanagihara and Satoh (2012) when $\lambda\vK=\0_q\0_q'$ and $\vX=\vI_p$. 

Moreover, from Theorems 2.3 and 2.4, and Corollary 2.4, we have the following corollary;
\begin{corollary}
If $\hat{\theta}_i^{[1]}(\lambda|\vS,\tilde{\theta}_i^{[0]}) \le \hat{\theta}_i^{[2]}(\lambda|\vS,\tilde{\theta}_i^{[1]})$ and $\hat{\theta}_i^{[s]}(\lambda|\vS,\tilde{\theta}_i^{[s-1]}) <\infty $ for any natural number $s \ge 2$, then $0 \le \hat{\theta}_i^{[1]}(\lambda|\vS,\tilde{\theta}_i^{[0]}) \le \hat{\theta}_i^{[2]}(\lambda|\vS,\tilde{\theta}_i^{[1]}) \le \cdots \le \hat{\theta}_i^{[s]}(\lambda|\vS,\tilde{\theta}_i^{[s-1]}) \le \hat{\theta}_i^{[s+1]}(\lambda|\vS,\tilde{\theta}_i^{[s]}) \le \cdots \le \hat{\theta}_i^{[\infty]}(\lambda|\vS)$ for each $i=1,\ldots,k$.
\end{corollary}
This is similar result for that of Nagai, Yanagihara and Satoh (2012).

From \tref{t3} to here, we often assume $\hat{\theta}_i^{[2]}(\lambda|\vS,\tilde{\theta}_i^{[1]})>\hat{\theta}_i^{[1]}(\lambda|\vS,\tilde{\theta}_i^{[0]})$.
We can see this inequation holds when $\lambda=0$ without any conditions;
\begin{theorem}
When $\lambda=0$, $\hat{\theta}_i^{[2]}(0|\vS,\tilde{\theta}_i^{[0]})>\hat{\theta}_i^{[1]}(0|\vS,\tilde{\theta}_i^{[0]})$ is always holding.
\label{t5}
\end{theorem}
\begin{proof}
Since $\hat{\beta}_{1,i}(0|\vS,\tilde{\theta}_i^{[s]})=\hat{\beta}_{2,i}(0|\vS,\tilde{\theta}_i^{[s]})$ from $\vG_0^2=\vG_0$ for any natural number $s$ and $\tr(\vS\vG_0\vS^{-1}\vG_0)>0$, we obtain $\hat{\theta}_i^{[1]}(0|\vS,\tilde{\theta}_i^{[0]}) \in (0,\infty)$ and $\hat{\theta}_i^{[2]}(0|\vS,\tilde{\theta}_i^{[1]}) \in (0,\infty)$.
Further, $\hat{\theta}_i^{[2]}(0|\vS,\tilde{\theta}_i^{[1]})$ is derived by $\{1+\hat{\theta}_i^{[1]}(0|\vS,\tilde{\theta}_i^{[0]}/d_i)\}^2\tr(\vS\vG_0\vS^{-1}\vG_0)/\ve_i'\vP\vY\vG_0\vS^{-1}\vG_0\vY'\vP'\ve_i$.
Here, we obtain $\hat{\theta}_i^{[2]}(0|\vS,\tilde{\theta}_i^{[0]})>\hat{\theta}_i^{[1]}(0|\vS,\tilde{\theta}_i^{[0]})$ for $\lambda=0$ since $\tr(\vS\vG_0\vS^{-1}\vG_0)/\ve_i'\vP\vY\vG_0\vS^{-1}\vG_0\vY'\vP'\ve_i>0$ and $d_i>0$.
\end{proof}
On the other hand, when $\lambda > 0$, there are some conditions for holding $\hat{\theta}_i^{[2]}(\lambda|\vS,\tilde{\theta}_i^{[1]})>\hat{\theta}_i^{[1]}(\lambda|\vS,\tilde{\theta}_i^{[0]})$;
\begin{theorem}
When $\lambda>0$ and $\hat{\theta}_i^{[1]}(\lambda|\vS,\tilde{\theta}_i^{[0]}) \in (0,\infty)$, the sufficient conditions for the inequation $\hat{\theta}_i^{[2]}(\lambda|\vS,\tilde{\theta}_i^{[1]})>\hat{\theta}_i^{[1]}(\lambda|\vS,\tilde{\theta}_i^{[0]})$ are $\ve_i'\vP\vY\vG_\lambda^2\vS^{-1}\vG_\lambda\vY'\vP'\ve_i>0$, $\ve_i'\vP\vY\vG_\lambda\{\vS^{-1}-\vG_\lambda\vS^{-1}\vG_\lambda\}\vJ_\lambda\vY'\vP'\ve_i>0$ and $\tr(\vS\vG_\lambda\vS^{-1}\vG_\lambda)\ve_i'\vP\vY\vG_\lambda\vS^{-1}(\vG_0+\vG_\lambda)\vJ_\lambda\vY'\vP'\ve_i+\ve_i'\vP\vY\vG_\lambda^2\vS^{-1}\vG_\lambda\vY'\vP'\ve_i'\vP\vY\vG_\lambda\vS^{-1}\vJ_\lambda\vY'\vP'\ve_i>\ve_i'\vP\vY\vG_\lambda\vS^{-1}\vG_0\vY'\vP'\ve_i\ve_i'\vP\vY\vG_\lambda^2\vS^{-1}\vG_\lambda\vJ_\lambda\vY'\vP'$, where $\vJ_\lambda=(\vG_0-\vG_\lambda)/\lambda=\vX(\vX'\vX)^{-1}\vK(\vX'\vX+\lambda\vK)^{-1}\vX'$.
\label{t6}
\end{theorem}
\begin{proof}
The condition $\ve_i'\vP\vY\vG_\lambda^2\vS^{-1}\vG_\lambda\vY'\vP'\ve_i>0$ means $\hat{\theta}_i^{[2]}(\lambda|\vS,\tilde{\theta}_i^{[1]})<\infty$.
Then, we firstly consider obtaining the condition $\hat{\theta}_i^{[2]}(\lambda|\vS,\tilde{\theta}_i^{[1]})>0$.

Firstly, we derive the condition for $\hat{\theta}_i^{[2]}(\lambda|\vS,\tilde{\theta}_i^{[1]})>0$ when $\hat{\theta}_i^{[1]}(\lambda|\vS,\tilde{\theta}_i^{[0]}) \in (0.\infty)$, since $\tilde{\theta}_i^{[0]}$ does not use $\lambda$ and $\vtheta$ although $\tilde{\theta}_i^{[1]}$ uses $\lambda$ and $\vtheta$.
From the condition for $\hat{\theta}_i^{[1]}(\lambda|\vS,\tilde{\theta}_i^{[0]}) \in (0.\infty)$, we note $\tr(\vS\vG_\lambda\vS^{-1}\vG_\lambda)>\ve_i'(\vP\vY\vG_\lambda\vS^{-1}\vG_0\vY'\vP'-\vP\vY\vG_\lambda\vS^{-1}\vG_\lambda\vY'\vP')\ve_i$.
Using this inequation and $\tr(\vS\vG_\lambda\vS^{-1}\vG_\lambda)>0$, we obtain 
\begin{align*}
&\left(1+\dfrac{\hat{\theta}_i^{[1]}(\lambda|\vS,\tilde{\theta}_i^{[0]}}{d_i}\right)^2\tr(\vS\vG_\lambda\vS^{-1}\vG_\lambda)+\ve_i'(\vP\vY\vG_\lambda^2\vS^{-1}\vG_\lambda^2\vY'\vP'-\vP\vY\vG_\lambda^2\vS^{-1}\vG_\lambda\vY'\vP')\ve_i\\
>&\tr(\vS\vG_\lambda\vS^{-1}\vG_\lambda)+\ve_i'(\vP\vY\vG_\lambda^2\vS^{-1}\vG_\lambda^2\vY'\vP'-\vP\vY\vG_\lambda^2\vS^{-1}\vG_\lambda\vY'\vP')\ve_i\\
>&\ve_i'(\vP\vY\vG_\lambda\vS^{-1}\vG_0\vY'\vP'-\vP\vY\vG_\lambda\vS^{-1}\vG_\lambda\vY'\vP')\ve_i+\ve_i'(\vP\vY\vG_\lambda^2\vS^{-1}\vG_\lambda^2\vY'\vP'-\vP\vY\vG_\lambda^2\vS^{-1}\vG_\lambda\vY'\vP')\ve_i\\
=&\lambda\ve_i'\vP\vY\vG_\lambda\{\vS^{-1}-\vG_\lambda\vS^{-1}\vG_\lambda\}\vJ_\lambda\vY'\vP'\ve_i,
\end{align*}
from using $\vG_\lambda\vG_0=\vG_\lambda$ and $\vG_0-\vG_\lambda=\lambda\vJ_\lambda$.
Then, when $\ve_i'\vP\vY\vG_\lambda\{\vS^{-1}-\vG_\lambda\vS^{-1}\vG_\lambda\}\vJ_\lambda\vY'\vP'\ve_i> 0$, we find that $(1+\hat{\theta}_i^{[1]}(\lambda|\vS,\tilde{\theta}_i^{[0]}/d_i)^2\tr(\vS\vG_\lambda\vS^{-1}\vG_\lambda)+\ve_i'(\vP\vY\vG_\lambda^2\vS^{-1}\vG_\lambda^2\vY'\vP'-\vP\vY\vG_\lambda^2\vS^{-1}\vG_\lambda\vY'\vP')\ve_i>0$.
This means that $\hat{\theta}_i^{[2]}(\lambda|\vS,\tilde{\theta}_i^{[1]}) \in (0,\infty)$.
Thus, the sufficient condition for $\hat{\theta}_i^{[2]}(\lambda|\vS,\tilde{\theta}_i^{[1]}) \in (0,\infty)$ is derived as $\ve_i'\vP\vY\vG_\lambda\{\vS^{-1}-\vG_\lambda\vS^{-1}\vG_\lambda\}\vJ_\lambda\vY'\vP'\ve_i> 0$.

Secodary, we derive the sufficien condition for $\hat{\theta}_i^{[2]}(\lambda|\vS,\tilde{\theta}_i^{[1]})>\hat{\theta}_i^{[1]}(\lambda|\vS,\tilde{\theta}_i^{[0]})$.
In order to obtain the condition for $\hat{\theta}_i^{[2]}(\lambda|\vS,\tilde{\theta}_i^{[1]})>\hat{\theta}_i^{[1]}(\lambda|\vS,\tilde{\theta}_i^{[0]})$, we calculate following way;
\begin{align*}
&\ve_i'\vP\vY\vG_\lambda^2\vS^{-1}\vG_\lambda\vY'\vP'\ve_i\ve_i'\vP\vY\vG_\lambda\vS^{-1}\vG_0\vY'\vP'\ve_i\{\hat{\theta}_i^{[2]}(\lambda|\vS,\tilde{\theta}_i^{[1]})-\hat{\theta}_i^{[1]}(\lambda|\vS,\tilde{\theta}_i^{[0]})\}\\
>&\{\tr(\vS\vG_\lambda\vS^{-1}\vG_\lambda)+\ve_i'\vP\vY\vG_\lambda^2\vS^{-1}\vG_\lambda^2\vY'\vP'\ve_i\}\ve_i'\vP\vY\vG_\lambda\vS^{-1}\vG_0\vY'\vP'\ve_i\\
&-\{\tr(\vS\vG_\lambda\vS^{-1}\vG_\lambda)+\ve_i'\vP\vY\vG_\lambda\vS^{-1}\vG_\lambda\vY'\vP'\ve_i\}\ve_i'\vP\vY\vG_\lambda^2\vS^{-1}\vG_\lambda\vY'\vP'\ve_i\\
=&\tr(\vS\vG_\lambda\vS^{-1}\vG_\lambda)\ve_i'\vP\vY\vG_\lambda\vS^{-1}(\vG_0-\vG_\lambda^2)\vY'\vP'\ve_i\\
&+\ve_i'\vP\vY\vG_\lambda^2\vS^{-1}\vG_\lambda^2\vY'\vP'\ve_i'\vP\vY\vG_\lambda\vS^{-1}\vG_0\vY'\vP'\ve_i-\ve_i'\vP\vY\vG_\lambda^2\vS^{-1}\vG_\lambda\vY'\vP'\ve_i\ve_i'\vP\vY\vG_\lambda\vS^{-1}\vG_\lambda\vY'\vP'\ve_i\\
=&\lambda\tr(\vS\vG_\lambda\vS^{-1}\vG_\lambda)\ve_i'\vP\vY\vG_\lambda\vS^{-1}(\vG_0+\vG_\lambda)\vJ_\lambda\vY'\vP'\ve_i\\
&+\ve_i'\vP\vY\vG_\lambda^2\vS^{-1}\vG_\lambda(\vG_\lambda-\vG_0)\vY'\vP'\ve_i'\vP\vY\vG_\lambda\vS^{-1}\vG_0\vY'\vP'\ve_i\\
&+\ve_i'\vP\vY\vG_\lambda^2\vS^{-1}\vG_\lambda\vY'\vP'\ve_i'\vP\vY\vG_\lambda\vS^{-1}(\vG_0-\vG_\lambda)\vY'\vP'\ve_i\\
=&\lambda\tr(\vS\vG_\lambda\vS^{-1}\vG_\lambda)\ve_i'\vP\vY\vG_\lambda\vS^{-1}(\vG_0+\vG_\lambda)\vJ_\lambda\vY'\vP'\ve_i\\
&-\lambda\ve_i'\vP\vY\vG_\lambda^2\vS^{-1}\vG_\lambda\vJ_\lambda\vY'\vP'\ve_i'\vP\vY\vG_\lambda\vS^{-1}\vG_0\vY'\vP'\ve_i+\lambda\ve_i'\vP\vY\vG_\lambda^2\vS^{-1}\vG_\lambda\vY'\vP'\ve_i'\vP\vY\vG_\lambda\vS^{-1}\vJ_\lambda\vY'\vP'\ve_i
\end{align*}
from $\vG_0-\vG_\lambda^2=(\vG_0+\vG_\lambda)(\vG_0-\vG_\lambda)=\lambda(\vG_0+\vG_\lambda)\vJ_\lambda$ since $\vG_0=\vG_0^2$ and $\vG_0\vG_\lambda=\vG_\lambda\vG_0=\vG_\lambda$.
From $\ve_i'\vP\vY\vG_\lambda^2\vS^{-1}\vG_\lambda\vY'\vP'\ve_i\ve_i'\vP\vY\vG_\lambda\vS^{-1}\vG_0\vY'\vP'\ve_i>0$, the sufficient condition for $\hat{\theta}_i^{[2]}(\lambda|\vS,\tilde{\theta}_i^{[1]})>\hat{\theta}_i^{[1]}(\lambda|\vS,\tilde{\theta}_i^{[0]})$ is same as the last expression is positive.
Further, from $\lambda>0$, the sufficient condition is as $\tr(\vS\vG_\lambda\vS^{-1}\vG_\lambda)\ve_i'\vP\vY\vG_\lambda\vS^{-1}(\vG_0+\vG_\lambda)\vJ_\lambda\vY'\vP'\ve_i+\ve_i'\vP\vY\vG_\lambda^2\vS^{-1}\vG_\lambda\vY'\vP'\ve_i'\vP\vY\vG_\lambda\vS^{-1}\vJ_\lambda\vY'\vP'\ve_i>\vP\vY\vG_\lambda\vS^{-1}\vG_0\vY'\vP'\ve_i\ve_i'\vP\vY\vG_\lambda^2\vS^{-1}\vG_\lambda\vJ_\lambda\vY'\vP'\ve_i'$.
\end{proof}
Here, we can regard that $\vG_\lambda\vS^{-1}\vG_\lambda$ is derived by shrinking $\vS^{-1}$.
From this reason, the second condition which for $\hat{\theta}_i^{[2]}(\lambda|\vS,\tilde{\theta}_i^{[1]}) \in (0,\infty)$ can be considered as natural condition.

We can see that some condition about the magnitude between $\hat{\theta}_i^{[2]}(\lambda|\vS,\tilde{\theta}_i^{[1]})$ and $\hat{\theta}_i^{[1]}(\lambda|\vS,\tilde{\theta}_i^{[0]})$ in several theorems and corollary can be rewrittein in the condition in \tref{t5} and \tref{t6}.


\vskip 8pt
\centerline{\bf \large 4. Numerical Studies}
\setcounter{section}{4}
\setcounter{equation}{0}
\vskip 8pt

In the present paper, we propose new optimization method.
By conducting simulations, we compare the proposed optimization method with the ordinary optimization method as $\hat{\vtheta}^{[\rm{\tiny C}]}(\lambda)=\arg\!\min_{\theta_i \ge 0}C_p(\vtheta,\lambda)$, $\hat{\vtheta}^{[\rm{\tiny M}]}(\lambda)=\arg\!\min_{\theta_i \ge 0}MC_p(\vtheta,\lambda)$, and LSEs.
Let $\vR_r=\diag(1,\ldots,r)$ and $\vDelta_r(\rho)$ be an $r \times r$ matrix whose $(i,j)$th element is $\rho^{|i-j|}$.

The explanatory matrix $\vA$ is given by $\vA=\vN_{n,k}\vPsi^{1/2}$ where $\vPsi=\vR_k^{1/2}\vDelta_k(\rho_a)\vR_k^{1/2}$ and $\vN_{u,v}$ is a $u \times v$ matrix whose each row vector is generated from the independent $v$-dimensional normal distribution with mean $\0_v$ and covariance matrix $\vI_v$.
Further, $\vXi$ is given by $\vN_{k,p}$ and $\vX$ is made by using {\it bs} function in {\it R} programing with {\it bs($1$:$p$,df=$q$)} which generates each value of $B$-spline function.

\if01
Moreover, let $\vm_i$ ($i=1,\ldots,12$) be a $p$-dimensional vector as follows;\\
\begin{tabular}{llll}
$\vm_1=\vh(\vt; e^2,e^{-1.5},e)$,&\!\!\!\! $\vm_2=\vh(\vt; e^2,e^{-1.5},e^2)$,&\!\!\!\! $\vm_3=\vh(\vt; e^2,e^{-2},e)$,&\!\!\!\! $\vm_4=\vh(\vt; e^2,e^{-2},e^2)$,\\[0.5mm]
$\vm_5=\vh(\vt; e^2,e^{-2.5},e)$,&\!\!\!\! $\vm_6=\vh(\vt; e^2,e^{-2.5},e^2)$,&\!\!\!\! $\vm_7=\vh(\vt; e^3,e^{-1.5},e)$,&\!\!\!\! $\vm_8=\vh(\vt; e^3,e^{-1.5},e)$,\\[0.5mm]
$\vm_9=\vh(\vt; e^3,e^{-2},e)$,&\!\!\!\! $\vm_{10}=\vh(\vt; e^3,e^{-2},e^2)$,&\!\!\!\! $\vm_{11}=\vh(\vt; e^3,e^{-2.5},e)$,&\!\!\!\! $\vm_{12}=\vh(\vt; e^3,e^{-2.5},e^2)$,
\end{tabular}
\\[1mm]
where $\vt=(1,\ldots,p)'$, and the $i$th element of $\vh(\vt;\alpha,\beta,\tau)$ is $\alpha\{1-\exp(-\beta t_i)\}^\tau$ which is the Richard's growth curve model (Richard, 1959).
Using these $\vm_i$ as $\vM_{k}(\vt)=(\vm_1,\ldots,\vm_k)'$, we make a true longitudinal trend by $\vA\vM(\vt)$.
\fi

We generate the response matrix $\vY$ from $N_{n\times p}(\vA\vXi\vX',\vSigma\otimes\vI_n)$ where $\vSigma=\vR_p^{1/2}\vDelta_p(\rho_y)\vR_p^{1/2}$, then we standardized $\vA$ in each column.
The penalty matrix $\vK$ is organized from $(-1) \times \vK_*$ where $\vK_*$ is set from the second difference matrix.
More details of $\vK_*$ and $B$-spline function are reported in Green and Silverman (1994).

When we use $\vtheta$ and $\lambda$, we search the optimize $\lambda$ in each methods by using ``optimize'' function which is a program in the {\it R} with $C_p$ and $MC_p$ criteria with optimized $\vtheta$ who depends on $\lambda$.

We simulate $10^3$ iteration for each $n$, $p$, $k$, $q$, $\rho_a$, and $\rho_y$.
In each iteration, $\vA$ and $\vXi$ are fixed but $\vY$ is varies, and optimize $\vtheta$ and $\lambda$ when we need. 
For each iteration, based on each optimized $\vtheta$ and $\lambda$, we derive $\tilde{\vY}$ that is derived from each estimation method.
Then, we evaluate the estimators by using the following function;
\begin{eqnarray*}
\tr\left\{\left(\tilde{\vY}-\vA\vXi\vX'\right)\vSigma^{-1}\left(\tilde{\vY}-\vA\vXi\vX'\right)'\right\},
\end{eqnarray*}
and we average it over iteration for each methods.
This average value can regard as the estimator for PMSE without constant values.

\begin{table}[htbp]
{\bf Table 1}; Fixed $n=30$ and changing other setting. The fist lines of table header without LSEs mean the optimized method for $\vtheta$.
The second lines of table header without LSEs mean the optimized criterion for $\lambda$.
In each row, bold font is the best result and italic font is the $2$nd best result.
Each values show rounded to two decimal places.
\begin{center}
\tabcolsep=1.1mm
\begin{tabular}{rrrrr|r|rr|rrrrrr} \hline\hline
\multicolumn{5}{c|}{Setting} & \multicolumn{1}{c|}{LSEs} & \multicolumn{1}{c}{$C_p$} & \multicolumn{1}{c|}{$MC_p$} & \multicolumn{2}{c}{$s=1$} & \multicolumn{2}{c}{$s=5$} & \multicolumn{2}{c}{$s\to\infty$} \\
\multicolumn{1}{c}{$p$} & \multicolumn{1}{c}{$k$} & \multicolumn{1}{c}{$q$} & \multicolumn{1}{c}{$\rho_a$} & \multicolumn{1}{c|}{$\rho_y$} & \multicolumn{1}{c|}{} & \multicolumn{1}{c}{$C_p$} & \multicolumn{1}{c|}{$MC_p$} & \multicolumn{1}{c}{$C_p$} & \multicolumn{1}{c}{$MC_p$} & \multicolumn{1}{c}{$C_p$} & \multicolumn{1}{c}{$MC_p$} & \multicolumn{1}{c}{$C_p$} & \multicolumn{1}{c}{$MC_p$} \\\hline\hline
5 & 5 & 3 & 0.20  & 0.20  & 19.90 & 19.75 & 19.79 & {\bf 19.61} & {\it 19.62} & 19.63 & 19.64 & 19.63 & 19.64 \\
 &  &  &  & 0.80  & 29.63 & {\bf 25.34} & {\it 25.36} & 29.37 & 29.37 & 29.41 & 29.42 & 29.41 & 29.42 \\
 &  &  & 0.90  & 0.20  & 19.41 & 15.54 & 15.24 & 15.92 & 15.92 & 14.88 & 14.89 & {\bf 14.87} & {\it 14.87} \\
 &  &  &  & 0.80  & 28.58 & {\it 19.08} & {\bf 18.77} & 23.03 & 23.03 & 21.22 & 21.22 & 21.05 & 21.05 \\
 &  &  & 0.99  & 0.20  & 19.66 & 12.57 & 11.94 & 13.80 & 13.80 & 11.33 & 11.33 & {\it 11.08} & {\bf 11.07} \\
 &  &  &  & 0.80  & 29.23 & {\it 15.95} & {\bf 15.38} & 20.28 & 20.28 & 16.71 & 16.71 & 16.31 & 16.31 \\\cline{2-14}
 & 10 & 3 & 0.20  & 0.20  & {\bf 36.05} & 36.88 & 37.22 & {\it 36.21} & 36.21 & 36.69 & 36.70 & 37.16 & 37.16 \\
 &  &  &  & 0.80  & 54.22 & {\bf 48.49} & {\it 48.73} & 54.38 & 54.38 & 55.24 & 55.24 & 56.26 & 56.28 \\
 &  &  & 0.90  & 0.20  & {\bf 35.88} & 36.82 & 37.48 & {\it 36.05} & 36.05 & 37.67 & 37.67 & 39.24 & 39.23 \\
 &  &  &  & 0.80  & 53.40 & {\bf 47.34} & {\it 47.84} & 54.82 & 54.82 & 58.39 & 58.39 & 61.57 & 61.57 \\
 &  &  & 0.99  & 0.20  & 35.88 & {\it 28.49} & {\bf 28.00} & 28.65 & 28.65 & 28.52 & 28.52 & 30.68 & 30.68 \\
 &  &  &  & 0.80  & 54.28 & {\it 37.39} & {\bf 37.11} & 44.24 & 44.24 & 44.09 & 44.09 & 46.19 & 46.19 \\\hline
10 & 5 & 3 & 0.20  & 0.20  & 25.00 & 24.87 & 25.07 & {\bf 24.14} & 24.21 & {\it 24.20} & 24.28 & 24.21 & 24.28 \\
 &  &  &  & 0.80  & 33.07 & {\bf 29.28} & {\it 29.48} & 32.14 & 32.19 & 33.08 & 33.13 & 33.50 & 33.55 \\
 &  &  & 0.90  & 0.20  & 24.47 & 19.18 & 18.45 & 19.41 & 19.43 & 17.60 & 17.63 & {\bf 17.42} & {\it 17.43} \\
 &  &  &  & 0.80  & 32.42 & {\it 21.90} & {\bf 21.20} & 25.25 & 25.26 & 23.60 & 23.62 & 24.30 & 24.30 \\
 &  &  & 0.99  & 0.20  & 24.32 & 15.96 & 14.51 & 16.72 & 16.73 & 13.49 & 13.50 & {\bf 13.13} & {\it 13.13} \\
 &  &  &  & 0.80  & 32.61 & 18.39 & {\bf 17.03} & 22.17 & 22.18 & 17.83 & 17.84 & {\it 17.22} & 17.23 \\\cline{3-14}
 &  & 7 & 0.20  & 0.20  & 44.01 & 44.24 & 44.30 & {\bf 43.81} & 43.84 & {\it 43.81} & 43.84 & 43.81 & 43.84 \\
 &  &  &  & 0.80  & 55.08 & {\bf 51.26} & {\it 51.28} & 54.68 & 54.68 & 54.68 & 54.68 & 54.68 & 54.68 \\
 &  &  & 0.90  & 0.20  & 44.02 & 34.36 & 32.13 & 35.38 & 35.40 & 31.37 & 31.38 & {\bf 31.02} & {\it 31.02} \\
 &  &  &  & 0.80  & 55.11 & 40.47 & {\bf 38.63} & 44.72 & 44.72 & 40.58 & 40.58 & {\it 40.33} & 40.34 \\
 &  &  & 0.99  & 0.20  & 44.61 & 32.39 & 30.14 & 33.51 & 33.52 & {\bf 30.12} & {\it 30.14} & 31.33 & 31.32 \\
 &  &  &  & 0.80  & 55.57 & {\it 37.97} & {\bf 36.35} & 42.28 & 42.28 & 39.62 & 39.63 & 41.55 & 41.56 \\\cline{2-14}
 & 10 & 3 & 0.20  & 0.20  & 44.75 & 47.13 & 48.20 & {\bf 44.17} & {\it 44.21} & 46.32 & 46.35 & 48.19 & 48.25 \\
 &  &  &  & 0.80  & 60.30 & {\bf 55.38} & {\it 56.47} & 58.66 & 58.68 & 63.60 & 63.62 & 67.85 & 67.87 \\
 &  &  & 0.90  & 0.20  & 44.27 & 45.94 & 47.63 & {\bf 43.29} & {\it 43.30} & 124.53 & 124.54 & 127.70 & 127.71 \\
 &  &  &  & 0.80  & 59.50 & {\bf 52.18} & {\it 53.77} & 57.15 & 57.15 & 65.47 & 65.47 & 73.84 & 73.85 \\
 &  &  & 0.99  & 0.20  & 44.27 & 33.41 & 31.46 & 32.51 & 32.51 & {\bf 30.58} & {\it 30.58} & 31.63 & 31.63 \\
 &  &  &  & 0.80  & 59.53 & 36.30 & {\bf 33.54} & 41.92 & 41.92 & 35.97 & 35.97 & {\it 35.66} & 35.67 \\\cline{3-14}
 &  & 7 & 0.20  & 0.20  & 81.48 & 83.18 & 83.83 & {\bf 81.40} & {\it 81.43} & 5683.54 & 5683.56 & 5683.54 & 5683.57 \\
 &  &  &  & 0.80  & 101.28 & {\bf 97.50} & {\it 97.86} & 102.86 & 102.88 & 13941.28 & 13941.31 & 13941.28 & 13941.31 \\
 &  &  & 0.90  & 0.20  & {\bf 80.95} & 84.23 & 84.92 & {\it 82.58} & 82.59 & 82.60 & 82.61 & 82.60 & 82.61 \\
 &  &  &  & 0.80  & {\bf 101.25} & {\it 102.55} & 102.96 & 109.51 & 109.51 & 109.60 & 109.60 & 109.60 & 109.60 \\
 &  &  & 0.99  & 0.20  & 81.17 & 72.84 & 72.24 & {\bf 71.52} & {\it 71.53} & 30180.33 & 30180.34 & 30186.06 & 30186.06 \\
 &  &  &  & 0.80  & 100.39 & {\bf 91.05} & {\it 91.19} & 97.63 & 97.63 & 101.04 & 101.04 & 109.28 & 109.26 \\\hline\hline
\end{tabular}
\end{center}
\label{}
\end{table}

\begin{table}[bp]
{\bf Table 2}; Fixed $n=50$ and changing other setting. The fist lines of table header without LSEs mean the optimized method for $\vtheta$.
The second lines of table header without LSEs mean the optimized criterion for $\lambda$.
In each row, bold font is the best result and italic font is the $2$nd best result.
Each values show rounded to two decimal places.
\begin{center}
\tabcolsep=1.1mm
\begin{tabular}{rrrrr|r|rr|rrrrrr} \hline\hline
\multicolumn{5}{c|}{Setting} & \multicolumn{1}{c|}{LSEs} & \multicolumn{1}{c}{$C_p$} & \multicolumn{1}{c|}{$MC_p$} & \multicolumn{2}{c}{$s=1$} & \multicolumn{2}{c}{$s=5$} & \multicolumn{2}{c}{$s\to\infty$} \\
\multicolumn{1}{c}{$p$} & \multicolumn{1}{c}{$k$} & \multicolumn{1}{c}{$q$} & \multicolumn{1}{c}{$\rho_a$} & \multicolumn{1}{c|}{$\rho_y$} & \multicolumn{1}{c|}{} & \multicolumn{1}{c}{$C_p$} & \multicolumn{1}{c|}{$MC_p$} & \multicolumn{1}{c}{$C_p$} & \multicolumn{1}{c}{$MC_p$} & \multicolumn{1}{c}{$C_p$} & \multicolumn{1}{c}{$MC_p$} & \multicolumn{1}{c}{$C_p$} & \multicolumn{1}{c}{$MC_p$} \\\hline\hline
5 & 5 & 3 & 0.20  & 0.20  & 19.96 & 19.82 & 19.83 & {\bf 19.74} & 19.75 & 19.74 & 19.75 & {\it 19.74} & 19.75 \\
 &  &  &  & 0.80  & 29.12 & {\bf 24.75} & {\it 24.75} & 28.46 & 28.46 & 28.46 & 28.46 & 28.46 & 28.46 \\
 &  &  & 0.90  & 0.20  & 19.57 & 15.44 & 15.29 & 16.07 & 16.07 & {\bf 15.11} & {\it 15.11} & 15.18 & 15.16 \\
 &  &  &  & 0.80  & 29.90 & {\it 20.24} & {\bf 20.11} & 24.41 & 24.41 & 22.81 & 22.81 & 22.83 & 22.83 \\
 &  &  & 0.99  & 0.20  & 19.78 & {\it 14.05} & {\bf 13.90} & 14.88 & 14.88 & 14.22 & 14.22 & 15.13 & 15.13 \\
 &  &  &  & 0.80  & 28.98 & {\it 18.07} & {\bf 17.98} & 21.84 & 21.84 & 21.13 & 21.13 & 22.23 & 22.23 \\\cline{2-14}
 & 10 & 3 & 0.20  & 0.20  & {\bf 35.97} & 36.07 & 36.07 & {\it 36.00} & 36.00 & 36.00 & 36.00 & 36.00 & 36.00 \\
 &  &  &  & 0.80  & 54.27 & {\it 47.34} & {\bf 47.34} & 54.71 & 54.71 & 54.71 & 54.71 & 54.71 & 54.71 \\
 &  &  & 0.90  & 0.20  & {\bf 36.57} & 37.45 & 37.73 & {\it 37.15} & 37.15 & 38.63 & 38.63 & 40.50 & 40.50 \\
 &  &  &  & 0.80  & 54.49 & {\bf 49.53} & {\it 49.73} & 57.73 & 57.73 & 61.89 & 61.89 & 66.02 & 66.02 \\
 &  &  & 0.99  & 0.20  & 35.88 & 31.25 & 31.71 & {\it 30.52} & {\bf 30.52} & 36.84 & 36.84 & 45.11 & 45.12 \\
 &  &  &  & 0.80  & 53.57 & {\bf 42.42} & {\it 42.86} & 47.80 & 47.80 & 57.39 & 57.39 & 67.70 & 67.70 \\\hline
10 & 5 & 3 & 0.20  & 0.20  & 24.60 & 24.06 & 24.08 & {\bf 23.88} & 23.91 & {\it 23.89} & 23.91 & 23.89 & 23.91 \\
 &  &  &  & 0.80  & 32.66 & {\bf 28.57} & {\it 28.58} & 31.35 & 31.35 & 31.35 & 31.36 & 31.35 & 31.36 \\
 &  &  & 0.90  & 0.20  & 24.10 & 18.14 & 17.80 & 19.17 & 19.18 & 17.42 & 17.43 & {\bf 17.30} & {\it 17.31} \\
 &  &  &  & 0.80  & 32.67 & {\it 21.18} & {\bf 20.82} & 25.37 & 25.38 & 23.26 & 23.26 & 23.44 & 23.44 \\
 &  &  & 0.99  & 0.20  & 23.94 & 16.05 & {\bf 15.67} & 17.27 & 17.27 & {\it 15.80} & 15.81 & 16.07 & 16.07 \\
 &  &  &  & 0.80  & 32.39 & {\it 18.37} & {\bf 17.91} & 22.51 & 22.51 & 19.34 & 19.34 & 19.22 & 19.20 \\\cline{3-14}
 &  & 7 & 0.20  & 0.20  & 43.90 & 43.93 & 43.95 & 43.77 & 43.79 & {\it 43.77} & 43.79 & {\bf 43.77} & 43.79 \\
 &  &  &  & 0.80  & 54.39 & {\bf 50.34} & {\it 50.35} & 54.05 & 54.06 & 54.05 & 54.06 & 54.05 & 54.06 \\
 &  &  & 0.90  & 0.20  & 44.14 & 35.99 & {\bf 35.53} & 36.75 & 36.76 & {\it 35.62} & 35.62 & 36.22 & 36.21 \\
 &  &  &  & 0.80  & 54.25 & {\it 43.31} & {\bf 43.20} & 46.54 & 46.54 & 47.43 & 47.43 & 49.32 & 49.30 \\
 &  &  & 0.99  & 0.20  & 44.12 & 32.77 & 31.95 & 34.74 & 34.74 & 31.13 & 31.13 & {\bf 31.01} & {\it 31.01} \\
 &  &  &  & 0.80  & 55.12 & {\it 38.60} & {\bf 37.89} & 44.27 & 44.27 & 40.55 & 40.56 & 40.44 & 40.44 \\\cline{2-14}
 & 10 & 3 & 0.20  & 0.20  & 45.11 & 45.38 & 45.40 & 44.65 & 44.65 & {\bf 44.65} & 44.65 & {\it 44.65} & 44.65 \\
 &  &  &  & 0.80  & 61.07 & {\bf 53.82} & {\it 53.83} & 60.25 & 60.25 & 60.25 & 60.25 & 60.25 & 60.25 \\
 &  &  & 0.90  & 0.20  & 44.55 & 43.79 & 44.62 & {\bf 43.78} & {\it 43.78} & 48.22 & 48.23 & 52.93 & 52.93 \\
 &  &  &  & 0.80  & 60.04 & {\bf 50.56} & {\it 51.41} & 57.62 & 57.62 & 68.87 & 68.87 & 81.49 & 81.48 \\
 &  &  & 0.99  & 0.20  & 44.62 & {\bf 32.73} & {\it 32.77} & 33.49 & 33.49 & 37.95 & 37.95 & 43.06 & 43.06 \\
 &  &  &  & 0.80  & 60.13 & {\it 32.12} & {\bf 31.37} & 40.12 & 40.12 & 37.86 & 37.86 & 40.12 & 40.11 \\\cline{3-14}
 &  & 7 & 0.20  & 0.20  & {\bf 80.92} & 82.06 & 82.07 & {\it 81.57} & 81.58 & 81.57 & 81.58 & 81.57 & 81.58 \\
 &  &  &  & 0.80  & 101.53 & {\bf 98.21} & {\it 98.21} & 104.56 & 104.56 & 104.56 & 104.56 & 104.56 & 104.56 \\
 &  &  & 0.90  & 0.20  & {\bf 81.57} & 83.58 & 83.93 & {\it 83.16} & 83.16 & 84.19 & 84.19 & 86.52 & 86.52 \\
 &  &  &  & 0.80  & {\bf 100.66} & {\it 104.45} & 104.59 & 112.72 & 112.72 & 113.29 & 113.29 & 114.06 & 114.06 \\
 &  &  & 0.99  & 0.20  & 80.62 & 74.55 & 75.42 & {\bf 74.10} & {\it 74.11} & 79.00 & 79.00 & 88.84 & 88.85 \\
 &  &  &  & 0.80  & 101.18 & {\bf 100.05} & {\it 100.84} & 109.44 & 109.44 & 117.36 & 117.36 & 128.16 & 128.16 \\\hline\hline
\end{tabular}
\end{center}
\caption{}
\label{}
\end{table}

From these results, we note that we can use improved estimator with using $\vtheta$ and $\lambda$ in almost cases.
In specially, when $n$ is small, we can see $\vtheta$ and $\lambda$ have a more improvement.
On the other hand, when $n$ becomes large, LSEs ($\hat{\vmu}$ and $\hat{\vXi}$) are best method if $p$, $k$ and $q$ are also large.

When $\rho_y$ is small, the plug-in type methods are the best and $2$nd best methods in many cases.
Furthermore, although repeat plug-in methods with $s=5$ or $s\to \infty$ are sometimes better than $s=1$, when $n$ is small and $p$, $k$ and $q$ are large, these methods show some bad results.
The reason of this, we consider that we use the $\hat{\theta}_i^{[\infty]}(\lambda|\vS)$ ($i=1,\ldots,k$) while the condition for \tref{t3} or \tref{t4} does not hold.
Moreover, from these results, we can consider that we can select $s$, which is the number of repetitions, for each data in a similar manner with Nagai, Fukui and Yanagihara (2013), though we fix $s$ as $1$, $5$ and $\infty$ in this simulation. 

Focus on the optimization for $\lambda$, for plug-in type methods, $C_p$ criterion is sufficient criterion in almost cases.
However, for $C_p$ type optimization method for $\vtheta$, sometimes $MC_p$ criterion becomes better criterion for $\lambda$.

\vskip 8pt
\centerline{\bf \large 5. Conclusion}
\setcounter{section}{5}
\setcounter{equation}{0}
\vskip 8pt

In the present paper, we consider estimating the longitudinal trend and varying coefficient curves by using known basis function.
However, the ordinary estimation method by using LSEs with basis function occurs overfitting problem.
Addition to this, in the ordinary estimation method,  there is unstably problem when there is high correlation between the some columns in $\vA$.
In order to avoid these problems, Nagai (2011) proposed the penalized estimator by extending the method in Yanagihara, Nagai and Satoh (2012).
Nagai (2011) also proposed the $C_p$-type criteria, that are shown in Section 2, and the optimization method based on them.

On the other hand, in Yanagihara, Nagai and Satoh (2009), there are another methods proposed by Nagai, Yanagihara and Satoh (2012).
In the present paper, we extend their method into our estimator for the GMANOVA model for avoiding overfitting and unstably problems.
For extending their plug-in type optimization method, we calculate $\PMSE[\hat{\vY}_{\vtheta,\lambda}]$.
Then, we derive the optimal $\vtheta$ which minimizes PMSE is derived as \eref{Eq24} when $\lambda$ is given.
Based on it, we proposed the plug-in and repeat plug-in optimization methods by substituting the estimators into unknown matrix $\vSigma$ and $\vXi$.
We show some properties about the repeat plug-in optimization method under some condition.
Further, under the same condition, convergence of the repetition is shown at \tref{t4} as similar as Nagai, Yanagihara and Satoh (2012).

By conducting numerical studies as simulations, we compare the LSEs which are derived as $\hat{\vmu}$ and $\hat{\vXi}$ and the $C_p$-type criteria for optimizing $\vtheta$ with our proposed optimization methods.
When $n$ is small, our proposed penalized estimators is better than the LSEs in almost all cases.
In some situation, the plug-in type methods become best method.
Through the results, we consider the optimal number of repetitions may be exist.
Thus, we consider extending Nagai, Fukui and Yanagihara (2013).
Moreover, we note that the $C_p$ criterion and $MC_p$ criterion are similar results in our case.

\vskip 8pt
\centerline{\bf \large A.1. Proof of the expanded PMSE}
\setcounter{section}{1}
\setcounter{equation}{0}
\renewcommand{\thesection}{\Alph{section}}
\vskip 8pt

In this section, we show that $E_\vY[\tr\{(\hat{\vY}_{\vtheta,\lambda}-E_\vY[\vY])\vSigma^{-1}(\hat{\vY}_{\vtheta,\lambda}-E_\vY[\vY])'\}]$ is corresponding to $f(\vtheta,\lambda|\vSigma,\vXi,\vmu)$.
In addition to this, we show that the optimal $\vtheta$ is derived by minimizing \eref{Eq23}.
This method is similar as that in Nagai, Yanagihara and Satoh (2012).

Firstly, we decompose $(\vA,\1_n)$ as $\vP(\vA,\1_n)\vQ_1=\vL$ where $\vP$ and $\vQ_1$ are $n \times n$ and $(k+1) \times (k+1)$ orthogonal matrices, $\vQ_1=\begin{pmatrix} \vQ & \0_k \\ \0_k' & 1\end{pmatrix}$, and
 $\vL=(\diag(\sqrt{d_1},\ldots,\sqrt{d_k},\sqrt{n}),\0_{k+1}\0_{n-k-1}')'$ by using the singular value decomposition since $\vA'\1_n=\0_k$ and $\vQ'\vA'\vA\vQ=\vD=\diag(d_1,\ldots,d_k)$.
Letting $\vZ=\vP\vY$, $\vGamma=\vQ'\vXi$ and $\vmV=\vP\vmE$, we can rewrite the GMANOVA model in \eref{Eq13} as follows;
\begin{eqnarray}
\vZ=\vL\begin{pmatrix}\vGamma\\\vmu'\end{pmatrix}\vX'+\vmV.
\label{A1}
\end{eqnarray}
Here, we note that $\Cov[\rmvec(\vZ)]=\vSigma\otimes\vI_n$ since $\vP\vP'=\vI_n$.
Since $\vP(\vA,\1_n)\vQ_1=\vL$, we obtain  $\vP\vA\vQ=\vC=(\diag(\sqrt{d_1},\ldots,\sqrt{d_k}),\0_k\0_{n-k}')'$.
Using this result, $\hat{\vGamma}_{\vtheta,\lambda}=\vQ'\hat{\vXi}_{\vtheta,\lambda}=(\vD+\vTheta)^{-1}\vC'\vZ\vX(\vX'\vX+\lambda\vK)^{-1}$ is an estimator for $\vGamma$.
Then, since $\hat{\vZ}_{\vtheta,\lambda}=\vL(\hat{\vGamma}_{\vtheta,\lambda}',\hat{\vmu}_\lambda)'\vX'=\vP\hat{\vY}_{\vtheta,\lambda}$ and $E_\vY[\vY]=\vP E_\vZ[\vZ]$, we derive
$$E_\vY\!\left[\!\tr\!\left\{\!\left(\hat{\vY}_{\vtheta,\lambda}-E_\vY[\vY]\right)\vSigma^{-1}\left(\hat{\vY}_{\vtheta,\lambda}-E_\vY[\vY]\right)'\!\right\}\right]\!=\!E_\vZ\!\left[\!\tr\!\left\{\!\left(\hat{\vZ}_{\vtheta,\lambda}-E_\vZ[\vZ]\right)\vSigma^{-1}\left(\hat{\vZ}_{\vtheta,\lambda}-E_\vZ[\vZ]\right)'\!\right\}\right].$$
Thus, we consider minimizing the right hand of the above equation in order to derive \eref{Eq24} which minimizes the left hand of it.
Since $\hat{\vZ}_{\vtheta,\lambda}-E_\vZ[\vZ]=\vL\{(\hat{\vGamma}_{\vtheta,\lambda}',\hat{\vmu}_\lambda)'-(\vGamma',\vmu)'\}\vX'$ and the definition of $\vL$, letting $\vZ=(\vz_1,\ldots,\vz_n)'$ and $\vGamma=(\vgamma_1,\ldots,\vgamma_k)'$, 
\begin{align*}
\hat{\vZ}_{\vtheta,\lambda}-E_\vZ[\vZ]&=
	\begin{pmatrix}
	\vD^{1/2}\hat{\vGamma}_{\vtheta,\lambda}-\vD^{1/2}\vGamma\\
	\sqrt{n}\hat{\vmu}_\lambda'-\sqrt{n}\vmu'\\
	\0_{n-k-1}\0_q'
	\end{pmatrix}\vX'=
	\begin{pmatrix}
	\dfrac{d_1}{d_1+\theta_1}\vz_1'\vX(\vX'\vX+\lambda\vK)^{-1}-\sqrt{d_1}\vgamma_1'\\
	\vdots\\
	\dfrac{d_1}{d_1+\theta_1}\vz_k'\vX(\vX'\vX+\lambda\vK)^{-1}-\sqrt{d_k}\vgamma_k'\\
	\sqrt{n}(\hat{\vmu}_\lambda-\vmu)'\\
	\0_{n-k-1}\0_q'
	\end{pmatrix}\vX'.
\end{align*}
Thus, we derive the following results;
\begin{align}
&\tr\!\left\{\!\left(\hat{\vZ}_{\vtheta,\lambda}-E_\vZ[\vZ]\right)\vSigma^{-1}\left(\hat{\vZ}_{\vtheta,\lambda}-E_\vZ[\vZ]\right)'\!\right\}\nonumber\\
&=\sum_{i=1}^k\left(\dfrac{d_i}{d_i+\theta_i}(\vX'\vX+\lambda\vK')^{-1}\vX'\vz_i-\sqrt{d_i}\vgamma_i\right)'\vX'\vSigma^{-1}\vX\left(\dfrac{d_i}{d_i+\theta_i}(\vX'\vX+\lambda\vK')^{-1}\vX'\vz_i-\sqrt{d_i}\vgamma_i\right) \label{EqA2}\\
&\nonumber +n(\hat{\vmu}_\lambda-\vmu)'\vX'\vSigma^{-1}\vX(\hat{\vmu}_\lambda-\vmu).
\end{align}
Hence, minimizing the PMSE in \eref{Eq22} coincides with minimizing the expectation of \eref{EqA2}.

We consider calculating the expectation of \eref{EqA2}.
Since $\hat{\vmu}_\lambda=(\vX'\vX+\lambda\vK)^{-1}\vX'\vY'\1_n/n$, and we assume $\vA'\1_n=\0_k$, we derive $E[\hat{\vmu}_\lambda]=(\vX'\vX+\lambda\vK)^{-1}\vX'\vX\vmu$ and $E[\hat{\vmu}_\lambda'\vX'\vSigma^{-1}\vX\hat{\vmu}_\lambda]=E[\1_n'\vY\vG_\lambda\vSigma^{-1}$ $\vG_\lambda\vY'\1_n]/n^2=\tr(E[\vY'\1_n\1_n'\vY]\vG_\lambda\vSigma^{-1}\vG_\lambda)/n^2$, where $\vG_\lambda=\vX(\vX'\vX+\lambda\vK)^{-1}\vX'$ which is a symmetric matrix.
Further, since $\vY=(\vy_1,\ldots,\vy_n)'$ and $\vy_i \indep \vy_j$ ($i \neq j$),
\begin{align*}
E[\vY'\1_n\1_n'\vY]
=E\left[\left(\sum_{i=1}^n\vy_i\right)\left(\sum_{j=1}^n\vy_j'\right)\right]=\sum_{i=1}^n\left(E[\vy_i\vy_i']+\sum_{j\neq i}E[\vy_i]E[\vy_j]'\right).
\end{align*}
Since $\vSigma=E[(\vy_i-E[\vy_i])(\vy_i-E[\vy_i])']=E[\vy_i\vy_i']-E[\vy_i]E[\vy_i]'$ 
we derive
$$\sum_{i=1}^nE[\vy_i\vy_i']=n\vSigma+\sum_{i=1}^nE[\vy_i]E[\vy_i]'.$$
Furthermore, since $\vA'\1_n=\0_k$, we obtain 
\begin{align*}
\sum_{i=1}^n\sum_{j\neq i}E[\vy_i]E[\vy_j]'
&=\sum_{i=1}^n\left(\sum_{j=1}^nE[\vy_i]E[\vy_j]'-E[\vy_i]E[\vy_i]'\right)\\
&=\vX(\vXi',\vmu)\begin{pmatrix}\vA'\\\1_n'\end{pmatrix}\1_n\1_n'(\vA,\1_n)\begin{pmatrix}\vXi\\\vmu'\end{pmatrix}\vX'-\sum_{i=1}^nE[\vy_i]E[\vy_i]'\\
&=\vX(\vXi',\vmu)\begin{pmatrix}\0_k \\ n\end{pmatrix}(\0_k',n)\begin{pmatrix}\vXi\\\vmu'\end{pmatrix}\vX'-\sum_{i=1}^nE[\vy_i]E[\vy_i]'\\
&=n^2\vX\vmu\vmu'\vX'-\sum_{i=1}^nE[\vy_i]E[\vy_i]'.
\end{align*}
Hence, $E[\vY'\1_n\1_n'\vY]=n\vSigma+n^2\vX\vmu\vmu'\vX'$ is derived.
Thus, we obtain
\begin{align*}
E[\hat{\vmu}_\lambda'\vX'\vSigma^{-1}\vX\hat{\vmu}_\lambda]
&=\dfrac{1}{n^2}\tr\left\{(n\vSigma+n^2\vX\vmu\vmu'\vX')\vG_\lambda\vSigma^{-1}\vG_\lambda\right\}\\
&=\dfrac{1}{n}\tr(\vSigma\vG_\lambda\vSigma^{-1}\vG_\lambda)+\vmu'\vX'\vG_\lambda\vSigma^{-1}\vG_\lambda\vX\vmu.
\end{align*}
Using these results, the expectation of $(\hat{\vmu}_\lambda-\vmu)'\vX'\vSigma^{-1}\vX(\hat{\vmu}_\lambda-\vmu)$ is derived as follows;
\begin{align}
&E[(\hat{\vmu}_\lambda-\vmu)'\vX'\vSigma^{-1}\vX(\hat{\vmu}_\lambda-\vmu)]\nonumber\\
&=\dfrac{1}{n}\tr(\vSigma\vG_\lambda\vSigma^{-1}\vG_\lambda)+\vmu'\vX'\vG_\lambda\vSigma^{-1}\vG_\lambda\vX\vmu-2\vmu'\vX'\vSigma^{-1}\vG_\lambda\vX\vmu+\vmu'\vX'\vSigma^{-1}\vX\vmu\nonumber\\
&=\dfrac{1}{n}\tr(\vSigma\vG_\lambda\vSigma^{-1}\vG_\lambda)+\vmu'\vX'(\vG_\lambda-\vI_p)\vSigma^{-1}(\vG_\lambda-\vI_p)\vX\vmu.
\label{EqA3}
\end{align}
Next, we calculate the expectation of the first term in right hand in \eref{EqA2}.
That is, we calculate 
\begin{align}
\sum_{i=1}^kE\!\left[\!\left(\dfrac{d_i}{d_i+\theta_i}(\vX'\vX+\lambda\vK')^{-1}\vX'\vz_i-\sqrt{d_i}\vgamma_i\!\right)'\!\vX'\vSigma^{-1}\vX\!\left(\dfrac{d_i}{d_i+\theta_i}(\vX'\vX+\lambda\vK')^{-1}\vX'\vz_i-\sqrt{d_i}\vgamma_i\right)\!\right].
\label{EqA4}
\end{align}
Here, we can obtain the following result;
\begin{eqnarray*}
&\left(\dfrac{d_i}{d_i+\theta_i}(\vX'\vX+\lambda\vK')^{-1}\vX'\vz_i-\sqrt{d_i}\vgamma_i\!\right)'\!\vX'\vSigma^{-1}\vX\!\left(\dfrac{d_i}{d_i+\theta_i}(\vX'\vX+\lambda\vK')^{-1}\vX'\vz_i-\sqrt{d_i}\vgamma_i\right)\\
&=\left(\dfrac{d_i}{d_i+\theta_i}\right)^2\vz_i'\vG_\lambda\vSigma^{-1}\vG_\lambda\vz_i-2\dfrac{\sqrt{d_i}d_i}{d_i+\theta_i}\vz_i'\vG_\lambda\vSigma^{-1}\vX\vgamma_i+d_i\vgamma_i'\vX'\vSigma^{-1}\vX\vgamma_i,
\end{eqnarray*}
where we recall $\vG_\lambda=\vX(\vX'\vX+\lambda\vK)^{-1}\vX'$ and it is a symmetric matrix.
In order to obtain \eref{EqA4}, we derive the expectation of the above result.

Recall $\ve_i$ is an $n$-dimensional vector whose only $i$th element is one and other elements are zeros.
Since $\vz_i'$ is the $i$th row of $\vZ$, we can express $\vz_i$ as $\vz_i'=\ve_i'\vZ$.
Further, since $E[\vZ]=\vL(\vGamma',\vmu)'\vX'$, we obtain $E[\vz_i'\vG_\lambda\vSigma^{-1}\vX\vgamma_i]=\ve_i'E[\vZ]\vG_\lambda\vSigma^{-1}\vX\vgamma_i=\sqrt{d_i}\vgamma_i'\vX'\vG_\lambda\vSigma^{-1}\vX\vgamma_i$.
Further, since $\Cov[\vz_i]=\vSigma$,  $E[\vz_i'\vG_\lambda\vSigma^{-1}\vG_\lambda\vz_i]=\tr(E[\vZ'\ve_i\ve_i'\vZ]\vG_\lambda\vSigma^{-1}\vG_\lambda)$, and  $E[\vZ'\ve_i\ve_i'\vZ]=\vSigma+d_i\vX\vgamma_i\vgamma_i'\vX'$, we derive 
\begin{align*}
E[\vz_i'\vG_\lambda\vSigma^{-1}\vG_\lambda\vz_i]
&=\tr(\vSigma\vG_\lambda\vSigma^{-1}\vG_\lambda)+d_i\vgamma_i'\vX'\vG_\lambda\vSigma^{-1}\vG_\lambda\vX\vgamma_i.
\end{align*}

Thus, \eref{EqA4} can be derived as follows;
\begin{align}
&\sum_{i=1}^kE\!\left[\!\left(\dfrac{d_i}{d_i+\theta_i}(\vX'\vX+\lambda\vK')^{-1}\vX'\vz_i-\sqrt{d_i}\vgamma_i\!\right)'\!\vX'\vSigma^{-1}\vX\!\left(\dfrac{d_i}{d_i+\theta_i}(\vX'\vX+\lambda\vK')^{-1}\vX'\vz_i-\sqrt{d_i}\vgamma_i\right)\!\right]\nonumber\\
&=\sum_{i=1}^k\left\{\left(\dfrac{d_i}{d_i+\theta_i}\right)^2\left(\tr(\vSigma\vG_\lambda\vSigma^{-1}\vG_\lambda)+d_i\vgamma_i'\vX'\vG_\lambda\vSigma^{-1}\vG_\lambda\vX\vgamma_i\right)-2\dfrac{d_i^2}{d_i+\theta_i}\vgamma_i'\vX'\vG_\lambda\vSigma^{-1}\vX\vgamma_i\right\}\nonumber\\
&\ \ \ +\sum_{i=1}^kd_i\vgamma_i'\vX'\vSigma^{-1}\vX\vgamma_i\nonumber\\
\begin{split}
	&=\sum_{i=1}^k\left\{\left(\dfrac{d_i}{d_i+\theta_i}\right)^2\left[\tr(\vSigma\vG_\lambda\vSigma^{-1}\vG_\lambda)+d_i\vgamma_i'\vX'\vG_\lambda\vSigma^{-1}\vG_\lambda\vX\vgamma_i\right]-\dfrac{2d_i^2}{d_i+\theta_i}\vgamma_i'\vX'\vG_\lambda\vSigma^{-1}\vX\vgamma_i\right\}\\[7mm]
&\ \ +\tr(\vGamma\vX'\vSigma^{-1}\!\vX\vGamma'\vD).
	\end{split}
\label{EqA5}
\end{align}
Substituting the results in \eref{EqA3} and \eref{EqA5} into the expectation of \eref{EqA2}, $E_\vY[\tr\{(\hat{\vY}_{\vtheta,\lambda}-E_\vY[\vY])\vSigma^{-1}$ $(\hat{\vY}_{\vtheta,\lambda}-E_\vY[\vY])'\}]=f(\vtheta,\lambda|\vSigma,\vXi,\vmu)$ is derived.
Here, from \eref{EqA3} and \eref{EqA5}, we note that only the first term of \eref{EqA5}, who coincides with $\varphi_i(\theta_i,\lambda|\vSigma,\vXi)$ in \eref{Eq23}, depends on $\vtheta$.

\vskip 8pt
\centerline{\bf \large A.2. Proof of \eref{Eq24}}
\setcounter{section}{2}
\setcounter{equation}{0}
\renewcommand{\thesection}{\Alph{section}}
\vskip 8pt

In this Section, we prove minimizing $\varphi_i(\theta_i,\lambda|\vSigma,\vXi)$ in \eref{Eq23} for fixed $\lambda$ is derived as \eref{Eq24}.
From above calculation, we can see that only the first term of \eref{EqA5} depends on $\theta_i$ ($i=1,\ldots,k$) in $\PMSE[\hat{\vY}_{\vtheta,\lambda}]$.
Thus, minimizing $\PMSE[\hat{\vY}_{\vtheta,\lambda}]$ is corresponding with minimizing the first term of \eref{EqA5}.
Recalling $\beta_{1,i}(\lambda|\vSigma,\vgamma_i)=d_i\vgamma_i'\vX'\vG_\lambda\vSigma^{-1}\vX\vgamma_i$, $\beta_{2,i}(\lambda|\vSigma,\vgamma_i)=d_i\vgamma_i'\vX'\vG_\lambda\vSigma^{-1}\vG_\lambda'\vX\vgamma_i$ and $\alpha_i=\tr(\vSigma\vG_\lambda\vSigma^{-1}\vG_\lambda')+\beta_{2,i}(\lambda|\vSigma,\vgamma_i)$, then $\varphi_i(\theta_i,\lambda|\vSigma,\vgamma_i)$ is written as follows;
\begin{align*}
\varphi_i(\theta_i,\lambda|\vSigma,\vgamma_i)=\left(\dfrac{d_i}{d_i+\theta_i}\!\right)^2\alpha_i(\lambda|\vSigma,\vgamma_i)-\dfrac{2d_i}{d_i+\theta_i}\beta_{1,i}(\lambda|\vSigma,\vgamma_i).
\end{align*}
For simple expression, we write $\varphi_i(\theta_i,\lambda|\vSigma,\vgamma_i)$ as $g(\theta_i)$ in this Section.
In order to obtain the optimal $\theta_i$ which minimizes PMSE, we different $g(\theta_i)$ with respect to $\theta_i$.
\begin{align*}
g'(\theta_i)=\dfrac{\partial g(\theta_i)}{\partial \theta_i}=\dfrac{2d_i}{(d_i+\theta_i)^3}\{(d_i+\theta_i)\beta_{1,i}(\lambda|\vSigma,\vgamma_i)-d_i\alpha_i(\lambda|\vSigma,\vgamma_i)\}.
\end{align*}
Here, $g'(\theta_i)|_{\theta_i=0}=2(\beta_{1,i}(\lambda|\vSigma,\vgamma_i)-\alpha_i(\lambda|\vSigma,\vgamma_i)\}/d_i$.
When we consider the solution of $g'(\theta_i)|_{\theta_i=\bar{\theta_i}}=0$ is derived as $\bar{\theta}_i=\infty$ and $d_i\{\alpha_1(\lambda|\vSigma,\vgamma_i)-\beta_{1,i}(\lambda|\vSigma,\vgamma_i)\}/\beta_{1,i}(\lambda|\vSigma,\vgamma_i)$.
We note that the sign of $g'(\theta_i)|_{\theta_i=0}$ is equal to $\beta_{1,i}(\lambda|\vSigma,\vgamma_i)-\alpha_i(\lambda|\vSigma,\vgamma_i)$.
Thus, we consider three conditions as $\beta_{1,i}(\lambda|\vSigma,\vgamma_i)>\alpha_i(\lambda|\vSigma,\vgamma_i)$ ($>0$), $\alpha_i(\lambda|\vSigma,\vgamma_i) \ge \beta_{1,i}(\lambda|\vSigma,\vgamma_i)>0$ and ($\alpha_i(\lambda|\vSigma,\vgamma_i)>$) $0 \ge \beta_{1,i}(\lambda|\vSigma,\vgamma_i)$.
Under the first condition, since ${\rm sign}(g'(\theta_i)|_{\theta_i=0})>0$, we derive the optimal $\theta_i$ is $0$.
Under the second condition, ${\rm sign}(g'(\theta_i)|_{\theta_i=0})\le 0$ and $\beta_{1,i}(\lambda|\vSigma,\vgamma_i)>0$, the optimal $\theta_i$ is derived from extreme value as $d_i\{\alpha_i(\lambda|\vSigma,\vgamma_i)-\beta_{1,i}(\lambda|\vSigma,\vgamma_i)\}/\beta_{1,i}(\lambda|\vSigma,\vgamma_i)$.
In the last condition $0\ge \beta_{1,i}(\lambda|\vSigma,\vgamma_i)$, since ${\rm sign}(g'(\theta_i)|_{\theta_i=0})<0$ from $\alpha_i(\lambda|\vSigma,\vgamma_i)>0$ and $\beta_{1,i}(\lambda|\vSigma,\vgamma_i)\le 0$, the optimal $\theta_i$ is derived as $\infty$.
These results show the optimal $\theta_i$ is obtained as \eref{Eq24}.

\if01
\vskip 8pt
\centerline{\bf \large A.3. Proof of Theorem 2.1}
\setcounter{section}{3}
\setcounter{equation}{0}
\renewcommand{\thesection}{\Alph{section}}
\vskip 8pt

In this subsection, we prove Theorem 2.1.
Theorem 2.1 means that $\hat{\theta}_i^{[s]}(\lambda|\vS,\hat{\vXi}_{\hat{\vtheta}^{[s-1]},\lambda})$ is monotone increasing with $s$ becomes large.
Firstly, we note that the $i$th row vector of $\hat{\vGamma}_{\vtheta,\lambda}$ is $\hat{\vgamma}_i(\theta_i,\lambda)'=\sqrt{d_i}\vz_i'\vX(\vX'\vX+\lambda\vK)^{-1}/(d_i+\theta_i)$.
Here, in order to simple expression, we write $\hat{\theta}_i^{[s]}(\lambda|\vS,\hat{\vXi}_{\hat{\vtheta}^{[s]},\lambda})$ as $\hat{\theta}_i^{[s]}(\lambda)$.
Using this, we derive
\begin{align*}
\hat{\theta}_i^{[s+1]}(\lambda)=\left\{\begin{tabular}{ll}
\multicolumn{2}{l}{$0$\ (if $\dfrac{d_i^2}{(d_i+\hat{\theta}_i^{[s]}(\lambda))^2}\vz_i'\vG_\lambda^2\vS^{-1}\vG_\lambda'\vz_i>\tr(\vS\vG_\lambda\vS^{-1}\vG_\lambda')+\dfrac{d_i^2}{(d_i+\hat{\theta}_i^{[s]}(\lambda))^2}\vz_i'\vG_\lambda^2\vS^{-1}\vG_\lambda^{'2}\vz_i$)}\\
\multicolumn{2}{l}{$\dfrac{\tr(\vS\vG_\lambda\vS^{-1}\vG_\lambda')+\{d_i/(d_i+\hat{\theta}_i^{[s]}(\lambda))\}^2\vz_i'\vG_\lambda^2\vS^{-1}(\vG_\lambda-\vI_p)'\vG_\lambda'\vz_i}{d_i\vz_i'\vG_\lambda^2\vS^{-1}\vG_\lambda'\vz_i/(d_i+\hat{\theta}_i^{[s]}(\lambda))^2}$} \\
\multicolumn{2}{r}{(if $\tr(\vS\vG_\lambda\vS^{-1}\vG_\lambda')\!+\!\dfrac{d_i^2}{(d_i\!+\!\hat{\theta}_i^{[s]}(\lambda))^2}\vz_i'\vG_\lambda^2\vS^{-1}\vG_\lambda^{'2}\vz_i\ge \dfrac{d_i^2}{(d_i\!+\!\hat{\theta}_i^{[s]}(\lambda))^2}\vz_i'\vG_\lambda^2\vS^{-1}\vG_\lambda'\vz_i>0$)}\\
$\infty$ & (if $0\ge \dfrac{d_i^2}{(d_i+\hat{\theta}_i^{[s]}(\lambda))^2}\vz_i'\vG_\lambda^2\vS^{-1}\vG_\lambda'\vz_i$)
\end{tabular}\right. \!\!\!.
\end{align*}

At first, we prove $\hat{\theta}_i^{[2]}(\lambda)\ge \hat{\theta}_i^{[1]}(\lambda)$.
If $0\ge d_i^2\vz_i'\vG_\lambda^2\vS^{-1}\vG_\lambda'\vz_i/\{(d_i+\hat{\theta}_i^{[1]}(\lambda))^2\}$, then $\hat{\theta}_i^{[2]}(\lambda)=\infty$.
This fact shows $\hat{\theta}_i^{[2]}(\lambda) \ge \hat{\theta}_i^{[1]}$.

Further, we consider $\tr(\vS\vG_\lambda\vS^{-1}\vG_\lambda')+d_i^2\vz_i'\vG_\lambda^2\vS^{-1}\vG_\lambda^{'2}\vz_i/(d_i+\hat{\theta}_i^{[1]}(\lambda))^2\ge d_i^2\vz_i'\vG_\lambda^2\vS^{-1}\vG_\lambda'\vz_i/(d_i+\hat{\theta}_i^{[1]}(\lambda))^2>0$ is satisfied, then $\hat{\theta}_i^{[2]}(\lambda)\ge 0$.
Since $d_i^2\vz_i'\vG_\lambda^2\vS^{-1}\vG_\lambda'\vz_i/(d_i+\hat{\theta}_i^{[1]}(\lambda))^2>0$, we note $\vz_i'\vG_\lambda^2\vS^{-1}\vG_\lambda'\vz_i>0$.
This means that $\hat{\theta}_i^{[1]}(\lambda)<\infty$.
On the other hand, if $\vz_i'\vG_\lambda^2\vS^{-1}\vG_\lambda'\vz_i> \tr(\vS\vG_\lambda\vS^{-1}\vG_\lambda')+\vz_i'\vG_\lambda^2\vS^{-1}\vG_\lambda'^2\vz_i$, then $\hat{\theta}_i^{[1]}(\lambda)=0\le \hat{\theta}_i^{[2]}(\lambda)$.
Next, we consider the condition $\tr(\vS\vG_\lambda\vS^{-1}\vG_\lambda')+\vz_i'\vG_\lambda^2\vS^{-1}\vG_\lambda^{'2}\vz_i\ge \vz_i'\vG_\lambda^2\vS^{-1}\vG_\lambda'\vz_i>0$ is also satisfied.
Then, we derive 
\begin{align*}
&\hat{\theta}_i^{[2]}(\lambda)-\hat{\theta}_i^{[1]}(\lambda)\\
=&\left(\dfrac{1}{d_i}\dfrac{d_i}{(d_i+\hat{\theta}_i^{[1]})^2}\vz_i'\vG_\lambda^2\vS^{-1}\vG_\lambda'\vz_i\right)^{-1}\\
&\left[\left\{\dfrac{1}{d_i}\tr(\vS\vG_\lambda\vS^{-1}\vG_\lambda')+\dfrac{d_i}{(d_i+\hat{\theta}_i^{[1]})^2}\vz_i'\vG_\lambda^2\vS^{-1}(\vG_\lambda-\vI_p)'\vG_\lambda'\vz_i\right\}\right.\\
&\left.\qquad-\left\{\dfrac{d_i}{(d_i+\hat{\theta}_i^{[1]})^2}\tr(\vS\vG_\lambda\vS^{-1}\vG_\lambda')+\dfrac{d_i}{(d_i+\hat{\theta}_i^{[1]})^2}\vz_i'\vG_\lambda^2\vS^{-1}(\vG_\lambda-\vI_p)'\vG_\lambda'\vz_i\right\}\right]\\
=&\left(\dfrac{1}{d_i}\dfrac{d_i}{(d_i+\hat{\theta}_i^{[1]})^2}\vz_i'\vG_\lambda^2\vS^{-1}\vG_\lambda'\vz_i\right)^{-1}\dfrac{1}{d_i}\left\{1-\dfrac{d_i^2}{(d_i+\hat{\theta}_i^{[1]})^2}\right\}\tr(\vS\vG_\lambda\vS^{-1}\vG_\lambda').
\end{align*}
Since $\hat{\theta}_i^{[1]}\ge 0$, we derive $d_i^2/(d_i+\hat{\theta}_i^{[1]})^2\le 1$.
Then, we obtain $\hat{\theta}_i^{[2]}(\lambda)\ge \hat{\theta}_i^{[1]}(\lambda)$ since $d_i>0$, $\vz_u'\vG_\lambda^2\vS^{-1}\vG_\lambda\vz_i>0$ from the condition, and $\tr(\vS\vG_\lambda\vS^{-1}\vG_\lambda')$ from $\vS$ is a positive definite matrix.

Lastly, we consider $\hat{\theta}_i^{[2]}=0$.
That is, $d_i^2\vz_i'\vG_\lambda^2\vS^{-1}\vG_\lambda'\vz_i/(d_i+\hat{\theta}_i^{[1]}(\lambda))^2>\tr(\vS\vG_\lambda\vS^{-1}\vG_\lambda')+d_i^2\vz_i'\vG_\lambda^2\vS^{-1}\vG_\lambda^{'2}\vz_i/(d_i+\hat{\theta}_i^{[1]}(\lambda))^2$ is satisfied.
This condition shows that the following condition is also satisfied; $$\vz_i'\vG_\lambda^2\vS^{-1}\vG_\lambda'\vz_i\!-\!\vz_i'\vG_\lambda^2\vS^{-1}\vG_\lambda^{'2}\vz_i>\dfrac{d_i^2}{(d_i+\hat{\theta}_i^{[1]}(\lambda))^2}\!\left\{\vz_i'\vG_\lambda^2\vS^{-1}\vG_\lambda'\vz_i\!-\!\vz_i'\vG_\lambda^2\vS^{-1}\vG_\lambda^{'2}\vz_i\right\}\!>\tr(\vS\vG_\lambda\vS^{-1}\vG_\lambda'),$$
since $\hat{\theta}_i^{[1]}(\lambda)>0$ and $\tr(\vS\vG_\lambda\vS^{-1}\vG_\lambda)>0$.
Thus, if this condition is satisfied, $\hat{\theta}_i^{[1]}(\lambda)=0=\hat{\theta}_i^{[2]}(\lambda)$.

Hence, we obtain $\hat{\theta}_i^{[2]}(\lambda)\ge \hat{\theta}_i^{[1]}(\lambda)$ under all situations.
Next, using this result and mathematical induction, we prove $\hat{\theta}_i^{[s+1]}(\lambda)\ge\hat{\theta}_i^{[s]}(\lambda)$ for any natural number $s$.

First, when $\hat{\theta}_i^{[s+1]}(\lambda)=\infty$, $\hat{\theta}_i^{[s+1]}(\lambda)\ge \hat{\theta}_i^{[s]}(\lambda)$ is directly derived.
From here, we assume $\hat{\theta}_i^{[s]}(\lambda)\ge \hat{\theta}_i^{[s-1]}(\lambda)$ for some $s \ge 2$.
Then, we consider under $0<\hat{\theta}_i^{[s+1]}(\lambda)<\infty$ situation.
This situation means that  $d_i^2\vz_i'\vG_\lambda^2\vS^{-1}\vG_\lambda'\vz_i/(d_i+\hat{\theta}_i^{[s]}(\lambda))^2>0$ is satisfied.
This means $\vz_i'\vG_\lambda^2\vS^{-1}\vG_\lambda'\vz_i>0$ and then $d_i^2\vz_i'\vG_\lambda^2\vS^{-1}\vG_\lambda'\vz_i/(d_i+\hat{\theta}_i^{[s-1]}(\lambda))^2>0$ since $d_i^2>0$ and $\hat{\theta}_i^{[s-1]}(\lambda)\ge0$.
That is $\hat{\theta}_i^{[s]}(\lambda)<\infty$.
It is clearly that $\hat{\theta}_i^{[s+1]}(\lambda) \ge \hat{\theta}_i^{[s]}(\lambda)$ when $\hat{\theta}_i^{[s]}(\lambda)=0$.
When $0<\hat{\theta}_i^{[s]}(\lambda)<\infty$ under $0<\hat{\theta}_i^{[s+1]}(\lambda)<\infty$, we calculate $\hat{\theta}_i^{[s+1]}(\lambda)-\hat{\theta}_i^{[s]}(\lambda)$ as follows;
\begin{align*}
&\left(\dfrac{d_i}{(d_i+\hat{\theta}_i^{[s]}(\lambda))^2}\dfrac{d_i}{(d_i+\hat{\theta}_i^{[s-1]}(\lambda))^2}\vz_i'\vG_\lambda^2\vS^{-1}\vG_\lambda'\vz_i\right)\left(\hat{\theta}_i^{[s+1]}(\lambda)-\hat{\theta}_i^{[s]}(\lambda)\right)\\
=&\left[\dfrac{d_i}{(d_i+\hat{\theta}_i^{[s-1]}(\lambda))^2}\tr(\vS\vG_\lambda\vS^{-1}\vG_\lambda')+\dfrac{d_i}{(d_i+\hat{\theta}_i^{[s-1]}(\lambda))^2}\left(\dfrac{d_i}{d_i+\hat{\theta}_i^{[s]}(\lambda)}\right)^2\vz_i'\vG_\lambda^2\vS^{-1}(\vG_\lambda-\vI_p)'\vG_\lambda'\vz_i\right.\\
&\left.-\left\{\dfrac{d_i}{(d_i+\hat{\theta}_i^{[s]}(\lambda))^2}\tr(\vS\vG_\lambda\vS^{-1}\vG_\lambda')+\dfrac{d_i}{(d_i+\hat{\theta}_i^{[s]}(\lambda))^2}\left(\dfrac{d_i}{d_i+\hat{\theta}_i^{[s-1]}(\lambda)}\right)^2\vz_i'\vG_\lambda^2\vS^{-1}(\vG_\lambda-\vI_p)'\vG_\lambda'\vz_i\right\}\right]\\
=&\left(\dfrac{d_i}{(d_i+\hat{\theta}_i^{[s-1]}(\lambda))^2}-\dfrac{d_i}{(d_i+\hat{\theta}_i^{[s]}(\lambda))^2}\right)\tr(\vS\vG_\lambda\vS^{-1}\vG_\lambda').
\end{align*}
Moreover, we obtain 
\begin{align*}
&\dfrac{d_i}{(d_i+\hat{\theta}_i^{[s-1]}(\lambda))^2}-\dfrac{d_i}{(d_i+\hat{\theta}_i^{[s]}(\lambda))^2}\\
=&\dfrac{d_i}{(d_i+\hat{\theta}_i^{[s-1]}(\lambda))^2(d_i+\hat{\theta}_i^{[s]}(\lambda))^2}\left((d_i+\hat{\theta}_i^{[s]}(\lambda))^2-(d_i+\hat{\theta}_i^{[s-1]}(\lambda))^2\right)\\
=&\dfrac{d_i}{(d_i+\hat{\theta}_i^{[s-1]}(\lambda))^2(d_i+\hat{\theta}_i^{[s]}(\lambda))^2}(2d_i+\hat{\theta}_i^{[s]}(\lambda)+\hat{\theta}_i^{[s-1]}(\lambda))(\hat{\theta}_i^{[s]}(\lambda)-\hat{\theta}_i^{[s-1]}(\lambda)).
\end{align*}
Using the assumption $\hat{\theta}_i^{[s]}(\lambda)\ge\hat{\theta}_i^{[s-1]}(\lambda)$, we obtain the above term is nonnegative.
Thus, $\hat{\theta}_i^{[s+1]}(\lambda)\ge\hat{\theta}_i^{[s]}(\lambda)$ is derived since $d_i>0$, $\hat{\theta}_i^{[s]}(\lambda) \ge 0$, $\hat{\theta}_i^{[s-1]}(\lambda) \ge 0$, $\vz_i'\vG_\lambda^2\vS^{-1}\vG_\lambda'\vz_i>0$ from $0<\hat{\theta}_i^{[s+1]}(\lambda)<\infty$, and $\tr(\vS\vG_\lambda\vS^{-1}\vG_\lambda) > 0$.
Lastly, we consider $\hat{\theta}_i^{[s+1]}(\lambda)=0$.
This means that $d_i^2\vz_i'\vG_\lambda^2\vS^{-1}\vG_\lambda'\vz_i/(d_i+\hat{\theta}_i^{[s]}(\lambda))^2>\tr(\vS\vG_\lambda\vS^{-1}\vG_\lambda')+d_i^2\vz_i'\vG_\lambda^2\vS^{-1}\vG_\lambda^{'2}\vz_i/(d_i+\hat{\theta}_i^{[s]}(\lambda))^2$ is satisfied.
This condition is equal to $d_i^2\{\vz_i'\vG_\lambda^2\vS^{-1}\vG_\lambda'\vz_i-\vz_i'\vG_\lambda^2\vS^{-1}\vG_\lambda'^2\vz_i\}/(d_i+\hat{\theta}_i^{[s]}(\lambda))^2>\tr(\vS\vG_\lambda\vS^{-1}\vG_\lambda') (>0)$.
From the assumption $\hat{\theta}_i^{[s]}(\lambda)\ge\hat{\theta}_i^{[s-1]}(\lambda)$, we note that $d_i^2/(d_i+\hat{\theta}_i^{[s-1]}(\lambda))^2 \ge d_i^2/(d_i+\hat{\theta}_i^{[s]}(\lambda))^2$.
Thus, we obtain
\begin{align*}
\dfrac{d_i^2}{(d_i\!+\!\hat{\theta}_i^{[s-1]}(\lambda))^2}\{\vz_i'\vG_\lambda^2\vS^{-1}\vG_\lambda'\vz_i-\vz_i'\vG_\lambda^2\vS^{-1}\vG_\lambda'^2\vz_i\} \ge \dfrac{d_i^2}{(d_i\!+\!\hat{\theta}_i^{[s]}(\lambda))^2}\{\vz_i'\vG_\lambda^2\vS^{-1}\vG_\lambda'\vz_i-\vz_i'\vG_\lambda^2\vS^{-1}\vG_\lambda'^2\vz_i\}.
\end{align*}
Hence, $d_i^2\{\vz_i'\vG_\lambda^2\vS^{-1}\vG_\lambda'\vz_i-\vz_i'\vG_\lambda^2\vS^{-1}\vG_\lambda'^2\vz_i\}/(d_i+\hat{\theta}_i^{[s-1]}(\lambda))^2>\tr(\vS\vG_\lambda\vS^{-1}\vG_\lambda)$ is also satisfied.
This result means that $\hat{\theta}_i^{[s]}(\lambda)=0.$

Hence, $\hat{\theta}_i^{[s+1]}(\lambda)\ge\hat{\theta}_i^{[s]}(\lambda)$ if $\hat{\theta}_i^{[s]}(\lambda)\ge\hat{\theta}_i^{[s-1]}(\lambda)$ for some $s \ge 2$ under all situations.

Since $\hat{\theta}_i^{[2]}(\lambda)\ge\hat{\theta}_i^{[1]}(\lambda)$ and this above results, we proved Theorem 2.1 from mathematical induction.
\fi

\vskip 1pt
\begin{center}

\end{center}

\end{document}